\documentclass[11pt]{article}

\usepackage[a4paper,margin=1in]{geometry}
\usepackage{amsmath,amssymb,amsthm,mathtools,bm}
\usepackage{mathrsfs}
\usepackage{enumitem}
\usepackage{natbib}
\usepackage{hyperref}
\usepackage{booktabs}
\usepackage{microtype}

\hypersetup{
  colorlinks=true,
  citecolor=blue,
  linkcolor=blue,
  urlcolor=blue
}

\newcommand{\R}{\mathbb{R}}
\newcommand{\E}{\mathbb{E}}
\newcommand{\bE}{\mathbf{E}}
\newcommand{\bV}{\mathbf{V}}
\newcommand{\Prob}{\mathbb{P}}
\newcommand{\Var}{\mathrm{Var}}
\newcommand{\Cov}{\mathrm{Cov}}
\newcommand{\1}{\mathbf{1}}

\newcommand{\T}{\top}
\newcommand{\op}{\mathrm{op}}
\newcommand{\tr}{\mathrm{tr}}

\newcommand{\sign}{\mathrm{sign}}

\newcommand{\norm}[1]{\left\lVert #1 \right\rVert}
\newcommand{\maxnorm}[1]{\left\lVert #1 \right\rVert_{\max}}
\newcommand{\onenorm}[1]{\left\lVert #1 \right\rVert_{1}}
\newcommand{\twonorm}[1]{\left\lVert #1 \right\rVert_{2}}
\newcommand{\frobnorm}[1]{\left\lVert #1 \right\rVert_{F}}

\newcommand{\Vor}{\widehat{\mathbf V}^{\,\mathrm{or}}}
\newcommand{\Vdb}{\widehat{\mathbf V}^{\,DB}}

\newcommand{\Vsoft}{\widehat{\mathbf V}^{\,\mathrm{soft}}}
\newcommand{\Vhard}{\widehat{\mathbf V}^{\,\mathrm{hard}}}
\newcommand{\Vtrue}{\mathbf V}
\newcommand{\Gammak}{\mathbf \Gamma_k}
\newcommand{\GhatkD}{\widehat{\mathbf \Gamma}^{\,D}_{k}}

\newcommand{\Goorhatk}{\widehat{\mathbf \Gamma}^{\,U}_{k}}

\theoremstyle{plain}
\newtheorem{theorem}{Theorem}[section]
\newtheorem{assumption}{Assumption}[section]
\newtheorem{proposition}{Proposition}[section]
\newtheorem{lemma}{Lemma}[section]
\newtheorem{corollary}{Corollary}[section]

\theoremstyle{definition}
\newtheorem{definition}{Definition}[section]

\theoremstyle{remark}
\newtheorem{remark}{Remark}[section]

\title{\bf Difference-Based High-Dimensional Long-Run Covariance Matrix Estimation for Mean-shift Time Series}
\author{Yanhong Liu$^1$, Fengyi Song$^2$, and Long Feng$^2$\\ 
$^1$Guangzhou University, $^2$Nankai University}
\date{\today}

\begin{document}
\maketitle

\begin{abstract}
We consider estimation of high-dimensional long-run covariance matrices for time series with nonconstant means, a setting in which conventional estimators can be severely biased. To address this difficulty, we propose a difference-based initial estimator that is robust to a broad class of mean variations, and combine it with hard thresholding, soft thresholding, and tapering to obtain sparse long-run covariance estimators for high-dimensional data. We derive convergence rates for the resulting estimators under general temporal dependence and time-varying mean structures, showing explicitly how the rates depend on covariance sparsity, mean variation, dimension, and sample size. Numerical experiments show that the proposed methods perform favorably in high dimensions, especially when the mean evolves over time.
\end{abstract}

\section{Introduction}
\label{sec:introduction}

Let $\{\bm{X}_t\}_{t=1}^n$ be a $p$-dimensional time series. A fundamental second-order object in time series analysis is the long-run covariance matrix
\[
\mathbf V
:=
\sum_{k=-\infty}^{\infty}\mathbf \Gamma_k,\
\mathbf \Gamma_k:=\Cov(\bm{Z}_0,\bm{Z}_k),
\]
where $\{\bm{Z}_t\}$ denotes the latent stationary noise component. When the spectral density matrix $f_{\bm{Z}}(\omega)$ exists, $\mathbf V=2\pi f_{\bm{Z}}(0)$, so $\mathbf V$ is precisely the zero-frequency summary of temporal dependence. This matrix plays a central role in statistical inference for temporally dependent data, including inference for sample means and linear contrasts, change-point analysis, trend inference, simultaneous confidence bands, and many other procedures driven by partial sums or temporally aggregated quantities \citep{newey1987,andrews1991,chan2022mac,chan2022diff}. In low-dimensional settings, long-run variance or long-run covariance estimation has been extensively studied through heteroskedasticity and autocorrelation consistent (HAC) estimators and related kernel or bandwidth-based procedures \citep{newey1987,andrews1991}. 

Classical HAC methodology is typically developed under the implicit premise that the observed series can be globally centered without materially distorting the serial dependence structure. This paradigm is adequate when the mean is stable, but it can break down when the data contain smooth trends, abrupt level shifts, or multiple change points. In such settings, naive centering may contaminate sample autocovariances, thereby leading to seriously biased long-run variance estimates. To overcome this difficulty, \citet{chan2022mac} and \citet{chan2022diff} developed mean-structure-robust, difference-based estimators that remain valid under unknown trends and even a possibly divergent number of breaks. However, that literature is essentially low-dimensional and does not address the additional regularization required when the ambient dimension $p$ is comparable to or much larger than the sample size $n$.

Meanwhile, the high-dimensional covariance literature has largely evolved from the i.i.d.\ setting. It is now well understood that the sample covariance matrix is no longer a reliable estimator when $p$ is large relative to $n$, which has led to a rich body of work on regularized covariance estimation. Representative examples include hard thresholding \citep{bickel2008threshold}, generalized thresholding and shrinkage \citep{rothman2009generalized}, adaptive thresholding \citep{cai2011adaptive}, and tapering-based estimators \citep{cai2010optimal}. These methods exploit sparsity or approximate sparsity of the target covariance matrix and have become standard tools in modern high-dimensional inference. However, they are designed primarily for contemporaneous covariance matrices under independence, and therefore do not directly address the cumulative serial dependence encoded in $\mathbf V$.

A growing literature extends high-dimensional second-order estimation to temporally dependent observations, but the main target in that line of work is still the contemporaneous covariance matrix $\mathbf \Sigma=\Cov(\bm{Z}_t)$ or its inverse, rather than the long-run covariance matrix $\mathbf V$. For example, \citet{chen2013hdts} established rates for thresholded covariance and graphical-Lasso-type precision matrix estimation for high-dimensional stationary and locally stationary time series under functional dependence conditions. \citet{shu2019temporal} further developed generalized thresholding and precision matrix estimation for temporally dependent observations, allowing slowly decaying dependence and long memory through autocorrelation-based dependence summaries. Under heavy tails and mild moments, \citet{zhang2021robust} studied robust estimation of means, covariance matrices and precision matrices for high-dimensional time series. Although these papers substantially broaden the scope of high-dimensional covariance analysis under dependence, their inferential target remains $\mathbf \Sigma$ or $\mathbf \Sigma^{-1}$ rather than the long-run covariance matrix needed for inference on temporal aggregates.

Related frequency-domain work considers high-dimensional spectral density matrices and their inverses. \citet{sun2018large} proposed thresholded estimators for large spectral density matrices of Gaussian and linear processes. \citet{fiecas2019spectral} developed non-asymptotic theory for smoothed periodogram estimation and sparse inverse spectral density estimation. For locally stationary processes, \citet{zhangwu2021convergence} derived a systematic asymptotic theory for covariance and spectral density estimation in high dimensions. These papers are highly relevant because the spectral density matrix is the frequency-domain analogue of the covariance structure. Nevertheless, they do not directly solve the problem considered here, namely, estimating a sparse long-run covariance matrix in the time domain when the observed series may contain nonconstant and unknown mean structures.

There is also a more recent line of work that comes closer to the present problem by directly targeting zero-frequency or long-memory objects. In particular, \citet{baek2023local} developed thresholded and penalized local Whittle estimators for high-dimensional long-run variance and precision matrices under short- and long-range dependence, and \citet{zhang2025joint} studied joint estimation of multiple precision matrices for long-memory time series. These contributions demonstrate that high-dimensional long-run second-order estimation is both statistically meaningful and technically challenging. However, the existing methods in this direction are mainly formulated in the frequency domain under stationary models, whereas many applications in econometrics, finance, environmental studies, and neuroscience involve unknown trends, structural breaks, or other nonstationary mean dynamics. This creates a genuine methodological gap between mean-robust long-run variance estimation in classical time series and sparse regularization in modern high-dimensional statistics.

The goal of this paper is to fill this gap. We study \emph{difference-based high-dimensional long-run covariance matrix estimation} for multivariate time series. Our starting point is a difference-based pilot estimator motivated by \citet{chan2022diff}, which is intrinsically robust to a broad class of mean structures. We then combine this pilot estimator with three regularization schemes that are standard in high-dimensional covariance estimation, namely, hard thresholding, soft thresholding, and tapering. In this way, we bridge two strands of literature that have largely developed separately: on the one hand, mean-structure-robust long-run variance estimation for dependent data; on the other hand, sparse matrix regularization for high-dimensional covariance estimation.

The proposed framework has several attractive features. First, by using difference statistics instead of global centering, it remains stable when the observed series contains smooth mean variation or abrupt mean changes. Second, the subsequent regularization step exploits sparsity of the long-run covariance matrix and therefore adapts naturally to high-dimensional regimes. Third, our methodology is entirely time-domain, making it conceptually simple and computationally convenient, while still targeting the quantity that is directly relevant for inference based on partial sums and sample averages. From a theoretical perspective, we establish convergence properties for the hard-thresholded, soft-thresholded, and tapered estimators under suitable dependence and sparsity conditions. From a numerical perspective, our simulation results show that these procedures perform well in genuinely high-dimensional settings and remain effective when the mean structure is nonconstant, where classical HAC-type estimators can deteriorate substantially.


The rest of the paper is organized as follows. Section~\ref{sec:method} introduces the difference-based oracle and feasible long-run covariance estimators and the three regularization schemes. We also establishes the main concentration and convergence results. Section~\ref{sec:simulation} reports simulation studies under both constant and nonconstant mean structures. In Section \ref{sec:data}, we apply our method to a change-point inference problem for time series. Section~\ref{sec:conclusion} concludes the paper.

\section{Methods}\label{sec:method}

\subsection{Difference-based long-run covariance matrix estimation}
Let
\begin{equation*}
\bm{X}_t = \bm\mu_t + \bm{Z}_t, \qquad t=1,\ldots,n,
\label{eq:model}
\end{equation*}
where \(\bm{X}_t,\bm\mu_t,\bm{Z}_t \in \R^{p_n}\), \(p_n\) is allowed to diverge with \(n\), and \(\{\bm{Z}_t\}_{t\in\mathbb Z}\) is a zero-mean strictly stationary process. Define the lag-\(k\) autocovariance matrix of \(\bm{Z}_t\) by
\[
\Gammak := \E\left(\bm{Z}_0 \bm{Z}_k^\T\right), \qquad k\in\mathbb Z.
\]
The target long-run covariance matrix is
\begin{equation*}
\Vtrue := \sum_{k\in\mathbb Z} \Gammak \in \R^{p_n\times p_n},
\label{eq:targetV}
\end{equation*}
assuming the series converges entrywise (equivalently, in max norm under our assumptions).

Suppose that \(\boldsymbol{\mu}_1=\cdots=\boldsymbol{\mu}_n\). A large body of work has been devoted to estimating the long-run covariance matrix \(\Vtrue\), and the existing methods may be broadly grouped into three classical categories. The first category consists of subsampling-based estimators; see, for example, \cite{MeketonSchmeiser1984,Carlstein1986,SongSchmeiser1995,PolitisRomanoWolf1999,ChanYau2017b}. A representative example is the overlapping batch means (OBM) estimator,
\[
\widehat{\mathbf{\Sigma}}_{\mathrm{OBM}, n}
:=
\frac{\ell}{n-\ell+1}
\sum_{i=\ell}^n
\left(
\frac{1}{\ell}\sum_{j=i-\ell+1}^i \widehat{\boldsymbol{X}}_j
\right)^{\otimes 2},
\]
where \(\ell \in \mathbb{N}\cap(1,n)\) is the batch size, \(\widehat{\boldsymbol{X}}_i=\boldsymbol{X}_i-\overline{\boldsymbol{X}}_n\), and \(\overline{\boldsymbol{X}}_n=n^{-1}\sum_{i=1}^n \boldsymbol{X}_i\). And \(\boldsymbol{a}^{\otimes 2}:=\boldsymbol{a}\otimes \boldsymbol{a}=\boldsymbol{a}\boldsymbol{a}^\top\) for any \(\boldsymbol{a}\in\mathbb{R}^p\).

The second category is formed by kernel-based estimators; see \cite{newey1987,andrews1991,Politis2011}. Two widely used examples are the Bartlett and quadratic spectral (QS) estimators,
\[
\widehat{\boldsymbol{\Sigma}}_{\mathrm{Bart}, n}
:=
\sum_{k=-\ell}^{\ell} \operatorname{Bart}(k/\ell)\widehat{\boldsymbol{\Gamma}}_k,
\qquad
\widehat{\boldsymbol{\Sigma}}_{\mathrm{QS}, n}
:=
\sum_{k=-(n-1)}^{n-1} \operatorname{QS}(k/\ell)\widehat{\boldsymbol{\Gamma}}_k,
\]
where
\[
\widehat{\boldsymbol{\Gamma}}_k
:=
n^{-1}\sum_{i=|k|+1}^n
\widehat{\boldsymbol{X}}_i\widehat{\boldsymbol{X}}_{i-|k|}^{\top},
\]
\[
\operatorname{Bart}(t):=(1-|t|)\mathbb{I}(|t|\le 1),
\]
and
\[
\operatorname{QS}(t)
:=
\frac{25}{12\pi^2 t^2}
\left\{
\frac{\sin(6\pi t/5)}{6\pi t/5}
-\cos(6\pi t/5)
\right\}.
\]

The third category includes resampling-based procedures; see, for instance, \cite{Kunsch1989,PolitisRomano1994,PaparoditisPolitis2001,Lahiri2003}. More recently, another line of research has developed estimators constructed from orthonormal series; see \cite{Phillips2005,Sun2013}. Practical issues concerning the choice of kernels and orthonormal bases are discussed in \citet{LazarusLewisStockWatson2018}, while \citet{Muller2014} investigated inference in the presence of strong serial dependence.

When the mean structure is nonconstant, that is, \(\bm \mu_i\neq \bm \mu_j\) for some \(i\neq j\), the classical estimators of the long-run covariance matrix may fail to be valid. In particular, standard HAC-type estimators are generally biased unless the mean is constant, and existing modifications for nonstationary time series are typically tailored to specific forms of mean structure and are not equipped with an implementable optimal bandwidth selector. To address this difficulty, \citet{chan2022mac} proposed a mean-structure and autocorrelation consistent (MAC) estimator, which is fully nonparametric and can be computed in a single pass without explicitly estimating the trend or locating change points. The key idea is to use a variate-difference and bi-differencing construction to remove the effect of the unknown mean structure. Specifically, define
\[
\widehat{\mathbf{\Delta}}_k
=
\frac{1}{2(n-|k|+1)}
\sum_{i=|k|+1}^n
\left(\bm X_i-\bm X_{i-|k|}\right)^{\otimes 2},
\qquad |k|=0,1,\ldots,n-1.
\]
Let \(L_k=c_0\ell+c_1|k|\) with \(c_0,c_1>0\). Then the MAC estimator of \(\mathbf{V}\) is given by
\[
\widehat{\mathbf{V}}_{0,q,n}
=
\sum_{k=-\ell}^{\ell}
K_q\!\left(\frac{|k|}{\ell}\right)
\left(
\widehat{\mathbf{\Delta}}_{L_k}-\widehat{\mathbf{\Delta}}_k
\right),
\qquad
K_q(x)=(1-|x|^q)\mathbf{1}(|x|\le 1),
\]
which yields an estimator that is robust to unknown trends and multiple change points.

\cite{chan2022diff} proposed a general framework for the estimation of the long-run variance for time series with nonconstant means. We adopt the difference-based construction in the multivariate extension of \citet{chan2022diff}. Let \(m_n\in\mathbb N\) be the difference order and \(h_n\in\mathbb N\) the lag spacing. Let \(\mathbf d_n=(d_{n,0},\ldots,d_{n,m_n})^\T\) be a normalized difference sequence satisfying
\begin{equation}
\sum_{j=0}^{m_n} d_{n,j}=0,
\qquad
\sum_{j=0}^{m_n} d_{n,j}^2 =1.
\label{eq:diff_normalized}
\end{equation}
Define the vector-valued difference statistics
\begin{equation*}
\mathbf{D}_t
:= \sum_{j=0}^{m_n} d_{n,j}\,\bm{X}_{t-jh_n},
\qquad t=m_n h_n+1,\ldots,n.
\label{eq:D_t}
\end{equation*}

For each integer \(k\) with \(|k|<\ell_n\), where \(\ell_n\) is the kernel bandwidth, define the difference-based sample autocovariance matrix
\begin{equation*}
\GhatkD
:= \frac{1}{n}\sum_{t=m_n h_n+|k|+1}^{n} \mathbf{D}_t \mathbf{D}_{t-|k|}^\T.
\end{equation*}
Let \(K(\cdot)\) be an even kernel supported on \([-1,1]\). The initial difference-based long-run covariance matrix estimator is
\begin{equation*}
\Vdb
:= \sum_{|k|<\ell_n} K\!\left(\frac{k}{\ell_n}\right)\GhatkD.
\label{eq:VDB}
\end{equation*}

To retain the mean contribution explicitly, define
\begin{equation}
\mathbf{M}_t := \sum_{j=0}^{m_n} d_{n,j}\,\bm\mu_{t-jh_n},
\qquad
\mathbf{U}_t := \sum_{j=0}^{m_n} d_{n,j}\,\bm{Z}_{t-jh_n},
\label{eq:MU}
\end{equation}
so that \(\mathbf{D}_t = \mathbf{M}_t + \mathbf{U}_t\). Define the oracle autocovariance estimator (built from the noise-only difference process) by
\begin{equation*}
\Goorhatk
:= \frac{1}{n}\sum_{t=m_n h_n+|k|+1}^{n} \mathbf{U}_t \mathbf{U}_{t-|k|}^\T,
\qquad
\Vor
:= \sum_{|k|<\ell_n} K\!\left(\frac{k}{\ell_n}\right)\Goorhatk.
\label{eq:Vor}
\end{equation*}

The mean-only deterministic contribution is
\begin{equation}
\mathbf B_{\mu,n}
:=
\sum_{|k|<\ell_n} K\!\left(\frac{k}{\ell_n}\right)
\left[
\frac{1}{n}\sum_{t=m_n h_n+|k|+1}^{n} \mathbf{M}_t \mathbf{M}_{t-|k|}^\T
\right].
\label{eq:Bmu}
\end{equation}
Define the residual mean-induced random remainder
\begin{equation}
\mathbf R_{\mu,n}
:= \Vdb - \Vor - \mathbf B_{\mu,n}.
\label{eq:Rmu}
\end{equation}
By construction, \(\mathbf R_{\mu,n}\) consists of the cross terms involving \(\mathbf{M}_t\) and \(\mathbf{U}_t\). We do \emph{not} impose a changepoint structure on \(\bm\mu_t\); instead, \(\mathbf B_{\mu,n}\) and \(\mathbf R_{\mu,n}\) will enter the convergence rate explicitly.

We now state assumptions needed for the high-dimensional max-norm analysis. These assumptions are designed to (i) preserve the construction of the estimator detailed in Section~6.2 of \citet{chan2022diff}, (ii) keep the mean-induced effect explicit, and (iii) enable a high-dimensional union-bound upgrade of the entrywise oracle concentration.

\begin{assumption}[Kernel and difference sequence]
\label{ass:kernel_diff}
The kernel \(K\) is even, bounded, supported on \([-1,1]\), and \(K(0)=1\). There exist constants \(q_0>0\) and \(C_K>0\) such that
\[
|1-K(x)|\le C_K |x|^{q_0}, \qquad |x|\le 1.
\]
The difference sequence \(\{d_{n,j}\}_{j=0}^{m_n}\) is normalized as in equation \eqref{eq:diff_normalized}. Moreover,
\[
m_n h_n = o(n), \qquad \ell_n\to\infty, \qquad \ell_n=o(n).
\]
\end{assumption}

\begin{assumption}[Uniform oracle bias bound inherited from DB kernel theory]
\label{ass:oracle_bias}
There exists a constant \(C_B>0\), independent of \(n\) and \((r,s)\), such that
\begin{equation*}
\maxnorm{\E(\Vor)-\Vtrue}
\le
C_B\left\{\ell_n^{-q_0} + \frac{(m_n+1)h_n}{n}\right\}.
\label{eq:oracle_bias_bound}
\end{equation*}
\end{assumption}

\begin{assumption}[Uniform Bernstein-type concentration for oracle entries]
\label{ass:bernstein}
There exist constants \(c_1,c_2>0\) such that for all sufficiently large \(n\), all \(1\le r,s\le p_n\), and all \(t>0\),
\begin{equation*}
\Prob\!\left(
\left|
\Vor[r,s]-\E\Vor[r,s]
\right|>t
\right)
\le
2\exp\!\left[
-c_1 n \min\!\left\{
\frac{t^2}{\ell_n+m_n h_n},\,
\frac{t}{\ell_n+m_n h_n}
\right\}
\right].
\end{equation*}
\end{assumption}

\begin{remark}
Assumptions \ref{ass:oracle_bias} and \ref{ass:bernstein} are standard in high-dimensional analysis of dependent data. They typically follow under suitable combinations of weak temporal dependence, such as physical dependence or mixing conditions, together with mild tail assumptions, for example, sub-exponential or Bernstein-type moment bounds on the coordinates. We state these assumptions directly in the main text because the primary goal of this paper is to develop the max-norm error analysis and the associated thresholding theory for the difference-based estimator in a high-dimensional setting. For completeness, the appendix provides several sufficient conditions under which both assumptions are satisfied.
\end{remark}

\begin{assumption}[Explicit mean-induced remainder]
\label{ass:mean_remainder}
There exists a deterministic sequence \(r_{\mu,n}\ge 0\) such that
\begin{equation*}
\maxnorm{\mathbf R_{\mu,n}} = O_p(r_{\mu,n}),
\end{equation*}
where \(\mathbf R_{\mu,n}\) is defined in equation \eqref{eq:Rmu}.
\end{assumption}

\begin{remark}
Assumption~\ref{ass:mean_remainder} does \emph{not} require \(\bm\mu_t\) to follow a changepoint model. The mean effect is fully retained through the explicit deterministic bias matrix \(\mathbf B_{\mu,n}\) in equation \eqref{eq:Bmu} and the remainder \(\mathbf R_{\mu,n}\). A detailed discussion of mean shifts is provided in Section~\ref{sec:mean_explicit_controls}.
\end{remark}

Next, we state the Max-norm rate of the initial difference-based estimator $\Vdb$.
Define the overall rate
\begin{equation}
a_n
:=
\maxnorm{\mathbf B_{\mu,n}}
+ r_{\mu,n}
+ \ell_n^{-q_0}
+ \frac{(m_n+1)h_n}{n}
+ \sqrt{\frac{(\ell_n+m_n h_n)\log p_n}{n}}
+ \frac{(\ell_n+m_n h_n)\log p_n}{n}.
\label{eq:an_def}
\end{equation}

\begin{theorem}
\label{thm:maxnorm_rate}
Suppose Assumptions~\ref{ass:kernel_diff}--\ref{ass:mean_remainder} hold and \(\log p_n = o(n)\). Then
\begin{equation*}
\maxnorm{\Vdb-\Vtrue} = O_p(a_n),
\end{equation*}
where \(a_n\) is defined in \eqref{eq:an_def}.
\end{theorem}

\begin{corollary}[Bandwidth-balanced rate under negligible mean remainder]
\label{cor:balanced_rate}
Assume the conditions of Theorem~\ref{thm:maxnorm_rate}. In addition, suppose
\[
\maxnorm{\mathbf B_{\mu,n}} + r_{\mu,n}
= o\!\left(
\ell_n^{-q_0}
+
\sqrt{\frac{\ell_n \log p_n}{n}}
\right),
\qquad
m_n h_n \lesssim \ell_n,
\qquad
\frac{(m_n+1)h_n}{n}
=o\!\left(
\ell_n^{-q_0}
+
\sqrt{\frac{\ell_n \log p_n}{n}}
\right).
\]
Then
\[
\maxnorm{\Vdb-\Vtrue}
=
O_p\!\left(
\ell_n^{-q_0}
+\sqrt{\frac{\ell_n \log p_n}{n}}
+\frac{\ell_n \log p_n}{n}
\right).
\]
If moreover \(\ell_n \log p_n / n \to 0\) and
\[
\ell_n \asymp \left(\frac{n}{\log p_n}\right)^{1/(2q_0+1)},
\]
then
\begin{equation}
\maxnorm{\Vdb-\Vtrue}
=
O_p\!\left\{
\left(\frac{\log p_n}{n}\right)^{q_0/(2q_0+1)}
\right\}.
\label{eq:balanced_rate}
\end{equation}
\end{corollary}

\subsection{Sparse Estimators}

When the dimension is fixed, the difference-based (DB) estimator is effective for time series with nonconstant mean structures. In high-dimensional settings, however, the DB estimator may suffer from substantial bias and may fail to be invertible, much like the sample covariance matrix in classical high-dimensional problems. This motivates the use of sparse regularization under the assumption that the long-run covariance matrix is sparse. In this paper, we focus on three widely used regularization schemes: hard thresholding \citep{bickel2008threshold}, soft thresholding \citep{rothman2009generalized}, and tapering \citep{cai2010optimal}. We note that other high-dimensional covariance estimators developed in the literature could also be combined with the DB estimator in a similar manner.

Let \(\tau_n>0\) be a threshold level. Define the hard-thresholding estimator \(\Vhard=(\widehat V^{\mathrm{hard}}_{rs})\) by
\begin{equation}
\widehat V^{\mathrm{hard}}_{rs}
=
\begin{cases}
\Vdb[r,r], & r=s,\\[4pt]
\Vdb[r,s]\1\!\left(|\Vdb[r,s]|\ge \tau_n\right), & r\neq s.
\end{cases}
\label{eq:hard_threshold}
\end{equation}
Similarly, define the soft-thresholding estimator \(\Vsoft=(\widehat V^{\mathrm{soft}}_{rs})\) by
\begin{equation}
\widehat V^{\mathrm{soft}}_{rs}
=
\begin{cases}
\Vdb[r,r], & r=s,\\[4pt]
\sign(\Vdb[r,s])\bigl(|\Vdb[r,s]|-\tau_n\bigr)_+, & r\neq s.
\end{cases}
\label{eq:soft_threshold}
\end{equation}

\begin{assumption}[Weak row-sparsity class]
\label{ass:sparsity}
There exist \(\alpha\in[0,1)\), \(M>0\), and a sequence \(c_{n,p}\ge 1\) such that \(\Vtrue\) is symmetric and
\begin{equation*}
\max_{1\le r\le p_n}\sum_{s=1}^{p_n} |V_{rs}|^\alpha \le c_{n,p},
\qquad
\max_{1\le r\le p_n} V_{rr}\le M.
\end{equation*}
\end{assumption}

\begin{theorem}[Thresholded estimation under weak row-sparsity]
\label{thm:thresholding}
Suppose Theorem~\ref{thm:maxnorm_rate} and Assumption~\ref{ass:sparsity} hold. Let \(\tau_n=C_\tau a_n\) with a sufficiently large constant \(C_\tau>0\). Then, for the hard-thresholding estimator \(\Vhard\),
\begin{align}
\maxnorm{\Vhard-\Vtrue} &= O_p(a_n), \label{eq:thr_max}\\
\twonorm{\Vhard-\Vtrue} &= O_p\!\bigl(c_{n,p}a_n^{\,1-\alpha}\bigr), \label{eq:thr_op}\\
\frac{1}{p_n}\frobnorm{\Vhard-\Vtrue}^2 &= O_p\!\bigl(c_{n,p}a_n^{\,2-\alpha}\bigr). \label{eq:thr_frob}
\end{align}
In particular, if \(c_{n,p}a_n^{1-\alpha}\to 0\), then \(\Vhard\) is operator-norm consistent.
\end{theorem}

\begin{theorem}[Soft-thresholding under weak row-sparsity]
\label{thm:soft_thresholding}
Suppose Theorem~\ref{thm:maxnorm_rate} and Assumption~\ref{ass:sparsity} hold. Let \(\tau_n=C_\tau a_n\) with \(C_\tau>0\) sufficiently large, then, for the soft-thresholding estimator \(\Vsoft\),
\begin{align}
\maxnorm{\Vsoft-\Vtrue} &= O_p(a_n), \label{eq:soft_max}\\
\twonorm{\Vsoft-\Vtrue} &= O_p\!\bigl(c_{n,p}a_n^{\,1-\alpha}\bigr), \label{eq:soft_op}\\
\frac{1}{p_n}\frobnorm{\Vsoft-\Vtrue}^2 &= O_p\!\bigl(c_{n,p}a_n^{\,2-\alpha}\bigr). \label{eq:soft_frob}
\end{align}
\end{theorem}

Theorems \ref{thm:thresholding} and \ref{thm:soft_thresholding} establish the rates of convergence for the hard-thresholding and soft-thresholding estimators, respectively. These rates depend on several key features of the problem, most notably the sparsity level of the long-run covariance matrix, the magnitude of the mean shift, and the relationship between the sample size and the dimension. Roughly speaking, a sparser covariance structure facilitates more accurate estimation, while a larger mean shift deteriorates performance through its impact on the bias term. Meanwhile, the dimensionality affects the stochastic error through the usual high-dimensional scaling. Consequently, the overall convergence behavior is governed by the interplay among dependence sparsity, signal contamination from the mean shift, and the effective amount of sample information. These results clarify the regimes under which thresholded estimators remain reliable for estimating the long-run covariance matrix in high-dimensional settings.

\begin{definition}[Tapering weights and tapering estimator]
\label{def:tapering}
For an integer \(k_n\ge 1\), let \(\mathbf W^{(k_n)}=(w_{rs}^{(k_n)})_{1\le r,s\le p_n}\) be a symmetric tapering weight matrix satisfying:
\begin{enumerate}[label=(\roman*)]
\item \(0\le w_{rs}^{(k_n)}\le 1\) for all \(r,s\);
\item \(w_{rs}^{(k_n)}=1\) if \(|r-s|\le k_n/2\);
\item \(w_{rs}^{(k_n)}=0\) if \(|r-s|\ge k_n\);
\item for each row \(r\), the number of nonzero weights \(\#\{s:w_{rs}^{(k_n)}\neq 0\}\le C_w k_n\), where \(C_w\) is a universal constant.
\end{enumerate}
Define the tapering estimator by the Hadamard product
\begin{equation}
\mathbf V^{\mathrm{tap}} := \mathbf W^{(k_n)}\circ \Vdb.
\label{eq:taper_est}
\end{equation}
\end{definition}

To establish the convergence rate of the tapering estimator, we need the following assumption.
\begin{assumption}[Bandable target class]
\label{ass:bandable}
There exist \(\nu>0\), constants \(C_{\mathrm{bd}}, C_{\mathrm{bd},2}, M_1>0\), such that \(\Vtrue\) is symmetric and:
\begin{align}
\max_{1\le r\le p_n}\sum_{s=1}^{p_n}|V_{rs}| &\le M_1, \notag\\
\max_{1\le r\le p_n}\sum_{|r-s|>k}|V_{rs}| &\le C_{\mathrm{bd}}\,k^{-\nu}, \qquad \forall\, k\ge 1, \notag\\
\frac{1}{p_n}\sum_{r=1}^{p_n}\sum_{|r-s|>k}V_{rs}^2 &\le C_{\mathrm{bd},2}\,k^{-2\nu}, \qquad \forall\, k\ge 1. \label{eq:bandable_tail_l2}
\end{align}
\end{assumption}

\begin{theorem}[Tapering under bandability]
\label{thm:tapering}
Suppose Theorem~\ref{thm:maxnorm_rate}, Assumption~\ref{ass:bandable}, and Definition~\ref{def:tapering} hold. Then
\begin{equation}
\twonorm{\mathbf V^{\mathrm{tap}}-\Vtrue}
=
O_p\!\left(k_n a_n + k_n^{-\nu}\right).
\label{eq:taper_op}
\end{equation}
In addition,
\begin{equation}
\frac{1}{p_n}\frobnorm{\mathbf V^{\mathrm{tap}}-\Vtrue}^2
=
O_p\!\left(k_n a_n^2 + k_n^{-2\nu}\right).
\label{eq:taper_frob}
\end{equation}
Hence, if \(k_n a_n\to 0\) and \(k_n\to\infty\), the tapering estimator is operator-norm consistent. The balancing choice \(k_n\asymp a_n^{-1/(\nu+1)}\) yields
\[
\twonorm{\mathbf V^{\mathrm{tap}}-\Vtrue}
=
O_p\!\left(a_n^{\nu/(\nu+1)}\right).
\]
\end{theorem}

Theorem~\ref{thm:tapering} is complementary to thresholding: thresholding targets sparse matrices without an index-distance structure, whereas tapering exploits a natural ordering and bandability. Both are driven by the same initial max-norm rate \(a_n\), which in our framework explicitly carries the mean-induced component.

\subsection{Selection of tuning parameters}

The regularization parameters are selected in a data-driven manner through a blockwise validation procedure tailored to dependent observations. Since random splitting at the individual-observation level would destroy the temporal dependence structure, we instead repeatedly draw two non-overlapping contiguous blocks from the sample. For the $b$th repetition, denote by $I_b^{\mathrm{tr}}$ and $I_b^{\mathrm{va}}$ the training block and the validation block, respectively, and let $\hat{\mathbf V}^{(b)}_{\lambda}$ be the thresholded estimator of the difference-based long-run covariance matrix constructed from the training block with threshold level $\lambda$. To evaluate the quality of $\hat{\mathbf V}^{(b)}_{\lambda}$, we compute on the validation block a pilot estimator $\tilde{\mathbf V}^{(b)}$, and measure the discrepancy by the squared Frobenius loss
\[
L_b(\lambda)=\bigl\|\hat{\mathbf V}^{(b)}_{\lambda}-\tilde{\mathbf V}^{(b)}\bigr\|_F^2.
\]
Averaging over $B$ random double-block splits yields the validation criterion
\[
\mathrm{CV}_{\lambda}(\lambda)=\frac{1}{B}\sum_{b=1}^B L_b(\lambda),
\]
and the selected threshold parameter is defined by
\[
\hat{\lambda}=\arg\min_{\lambda\in\Lambda}\mathrm{CV}_{\lambda}(\lambda),
\]
where $\Lambda$ is a prespecified candidate set. This criterion formalizes the usual bias--variance trade-off: an excessively small $\lambda$ retains too many spurious small entries and leads to high variability, whereas an excessively large $\lambda$ overshrinks the covariance structure and induces substantial bias.

The tapering bandwidth $k_{\mathrm{taper}}$ is chosen according to the same principle. For each candidate $k\in\mathcal K$, let $\hat{\mathbf V}^{(b)}_{k}$ denote the tapering estimator computed from the training block with bandwidth $k$. Its validation loss is defined as
\[
L_b(k)=\bigl\|\hat{\mathbf V}^{(b)}_{k}-\tilde{\mathbf V}^{(b)}\bigr\|_F^2,
\]
and the corresponding average criterion is
\[
\mathrm{CV}_{k}(k)=\frac{1}{B}\sum_{b=1}^B L_b(k).
\]
We then select
\[
\hat{k}_{\mathrm{taper}}=\arg\min_{k\in\mathcal K}\mathrm{CV}_{k}(k).
\]
From a statistical perspective, $k_{\mathrm{taper}}$ controls the effective range of serial dependence retained in the estimator: a small bandwidth may remove non-negligible autocovariance contributions and thus incur truncation bias, while a large bandwidth tends to preserve more long-lag information at the cost of increased estimation variability. Therefore, the proposed blockwise validation rule provides a practically stable and theoretically reasonable calibration mechanism for both $\lambda$ and $k_{\mathrm{taper}}$, while respecting the dependence structure of the time series.


\subsection{Discussion of Means}
\label{sec:mean_explicit_controls}

The rate \(a_n\) in equation \eqref{eq:an_def} contains two mean-related components,
\[
\maxnorm{\mathbf B_{\mu,n}}
\quad\text{and}\quad
r_{\mu,n},
\]
where \(\mathbf B_{\mu,n}\) is the deterministic mean-induced bias matrix in \eqref{eq:Bmu} and \(r_{\mu,n}\) controls the stochastic cross-term remainder \(\mathbf R_{\mu,n}\) in \eqref{eq:Rmu}. We now provide concrete upper bounds under broad mean regularity classes.

Recall \(\mathbf{M}_t\) from \eqref{eq:MU}. Define
\begin{equation}
\bar M_{2,n}
:=
\frac{1}{n}\sum_{t=m_n h_n+1}^{n}\norm{\mathbf{M}_t}_{\max}^2,
\qquad
\bar M_{\infty,n}
:=
\max_{m_n h_n+1\le t\le n}\norm{\mathbf{M}_t}_{\max}.
\label{eq:Mbar_def}
\end{equation}
Also write \(\norm{K}_{\infty}:=\sup_{|x|\le 1}|K(x)|\) and
\[
L_n := \ell_n + m_n h_n
\]
for short.

\begin{proposition}[Generic deterministic bound for \(\mathbf B_{\mu,n}\)]
\label{prop:Bmu_generic}
Under Assumption~\ref{ass:kernel_diff}, the deterministic mean-induced bias matrix \(\mathbf B_{\mu,n}\) in \eqref{eq:Bmu} satisfies
\begin{equation*}
\maxnorm{\mathbf B_{\mu,n}}
\le
(2\ell_n-1)\norm{K}_{\infty}\,\bar M_{2,n}.
\end{equation*}
Hence \(\maxnorm{\mathbf B_{\mu,n}}=O(\ell_n \bar M_{2,n})\).
\end{proposition}

To control the stochastic mean-induced remainder \(\mathbf R_{\mu,n}\), we impose a conditional Bernstein inequality on the cross terms. This assumption is standard once \(\mathbf U_t\) admits suitable dependence and tail controls uniformly over coordinates, and \(\mathbf M_t\) is treated as deterministic conditioning information.

\begin{assumption}[Conditional Bernstein inequality for the mean-noise cross term]
\label{ass:cross_bernstein}
There exists a constant \(c_3>0\) such that, conditional on \(\{\bm\mu_t\}_{t=1}^n\), for all sufficiently large \(n\) and all \(t>0\),
\begin{equation*}
\Prob\!\left(
\maxnorm{\mathbf R_{\mu,n}} > t
\ \middle|\ \{\bm\mu_t\}_{t=1}^n
\right)
\le
4p_n^2 \exp\!\left[
-c_3 n \min\!\left\{
\frac{t^2}{L_n^2 \bar M_{2,n}},
\frac{t}{L_n \bar M_{\infty,n}}
\right\}
\right].
\end{equation*}
\end{assumption}

We also give a sufficient condition for Assumption \ref{ass:cross_bernstein} in the Appendix.

\begin{proposition}[Generic stochastic bound for \(\mathbf R_{\mu,n}\)]
\label{prop:Rmu_generic}
Suppose Assumption~\ref{ass:cross_bernstein} holds and \(\log p_n=o(n)\). Then
\begin{equation}
\maxnorm{\mathbf R_{\mu,n}}
=
O_p\!\left(
\sqrt{\bar M_{2,n}}\,L_n\sqrt{\frac{\log p_n}{n}}
+
\bar M_{\infty,n}\frac{L_n\log p_n}{n}
\right).
\label{eq:Rmu_generic_bound}
\end{equation}
Consequently, one may take
\begin{equation}
r_{\mu,n}
\asymp
\sqrt{\bar M_{2,n}}\,L_n\sqrt{\frac{\log p_n}{n}}
+
\bar M_{\infty,n}\frac{L_n\log p_n}{n}.
\label{eq:rmu_generic_choice}
\end{equation}
\end{proposition}

We next show how \(\bar M_{2,n}\) and \(\bar M_{\infty,n}\) can be bounded under two concrete mean classes.

\subsubsection{Piecewise-H\"older mean class (with possible jumps)}

\begin{assumption}[Piecewise-H\"older mean path]
\label{ass:mean_piecewise_holder}
There exist an integer \(J_n\ge 0\), a partition
\[
0=\eta_0 < \eta_1 < \cdots < \eta_{J_n} < \eta_{J_n+1}=1,
\]
a H\"older exponent \(\beta\in(0,1]\), and constants \(L_\mu,M_\mu>0\) such that for each coordinate \(r\in\{1,\ldots,p_n\}\),
\[
\mu_t^{(r)} = g_r(t/n), \qquad t=1,\ldots,n,
\]
where \(g_r\) may jump at \(\eta_1,\ldots,\eta_{J_n}\), satisfies
\[
\sup_{u\in[0,1]} |g_r(u)|\le M_\mu,
\]
and on each open interval \((\eta_j,\eta_{j+1})\),
\[
|g_r(u)-g_r(v)|\le L_\mu |u-v|^\beta,
\qquad \forall\, u,v\in(\eta_j,\eta_{j+1}).
\]
\end{assumption}

\begin{proposition}[Mean-bias and remainder under piecewise-H\"older means]
\label{prop:holder_mean_bounds}
Suppose Assumption~\ref{ass:mean_piecewise_holder} holds. Let
\[
C_{d,1}:=\sum_{j=0}^{m_n}|d_{n,j}|,
\qquad
C_{d,\beta}:=\sum_{j=1}^{m_n}|d_{n,j}|j^\beta.
\]
Then:
\begin{enumerate}[label=(\roman*)]
\item The difference-mean process satisfies
\begin{equation}
\bar M_{2,n}
\le
C\left\{
\left(\frac{h_n}{n}\right)^{2\beta} C_{d,\beta}^2
+
C_{d,1}^2\,M_\mu^2\,\frac{J_n m_n h_n}{n}
\right\},
\label{eq:Mbar2_holder}
\end{equation}
and
\begin{equation}
\bar M_{\infty,n}
\le
C\left\{
C_{d,\beta}\left(\frac{h_n}{n}\right)^\beta
+
C_{d,1}M_\mu
\right\}.
\label{eq:Mbarinf_holder}
\end{equation}

\item Consequently, by Proposition~\ref{prop:Bmu_generic},
\begin{equation}
\maxnorm{\mathbf B_{\mu,n}}
\le
C\ell_n\left\{
\left(\frac{h_n}{n}\right)^{2\beta} C_{d,\beta}^2
+
C_{d,1}^2 M_\mu^2 \frac{J_n m_n h_n}{n}
\right\}.
\label{eq:Bmu_holder}
\end{equation}

\item If Assumption~\ref{ass:cross_bernstein} also holds, then by Proposition~\ref{prop:Rmu_generic},
\begin{align}
\maxnorm{\mathbf R_{\mu,n}}
&=
O_p\!\Bigg(
\sqrt{\left\{
\left(\frac{h_n}{n}\right)^{2\beta} C_{d,\beta}^2
+
C_{d,1}^2 M_\mu^2 \frac{J_n m_n h_n}{n}
\right\}}
L_n\sqrt{\frac{\log p_n}{n}}
\notag\\
&\hspace{2.2cm}
+
\left\{
C_{d,\beta}\left(\frac{h_n}{n}\right)^\beta + C_{d,1}M_\mu
\right\}
\frac{L_n\log p_n}{n}
\Bigg).
\label{eq:Rmu_holder}
\end{align}
\end{enumerate}
Here \(C>0\) is a constant independent of \(n\) and \(p_n\).
\end{proposition}

\subsubsection{Bounded-variation mean class (no changepoint parametrization)}

\begin{assumption}[Coordinatewise discrete bounded variation]
\label{ass:mean_bv}
There exists \(M_\mu>0\) and a sequence \(V_{\mu,n}\ge 0\) such that
\[
\sup_{1\le t\le n}\norm{\bm\mu_t}_{\max}\le M_\mu
\]
and
\begin{equation*}
\max_{1\le r\le p_n}\sum_{t=2}^{n}\left|\mu_t^{(r)}-\mu_{t-1}^{(r)}\right|
\le V_{\mu,n}.
\end{equation*}
\end{assumption}

\begin{proposition}[Mean-bias and remainder under bounded-variation means]
\label{prop:bv_mean_bounds}
Suppose Assumption~\ref{ass:mean_bv} holds. Let \(C_{d,1}:=\sum_{j=0}^{m_n}|d_{n,j}|\). Then
\begin{equation}
\bar M_{2,n}
\le
C\,C_{d,1}^2\,(m_n h_n)^2\,\frac{V_{\mu,n}^2}{n},
\label{eq:Mbar2_bv}
\end{equation}
and
\begin{equation}
\bar M_{\infty,n}
\le 2 C_{d,1} M_\mu.
\label{eq:Mbarinf_bv}
\end{equation}
Consequently,
\begin{equation}
\maxnorm{\mathbf B_{\mu,n}}
\le
C\,\ell_n\,C_{d,1}^2\,(m_n h_n)^2\,\frac{V_{\mu,n}^2}{n},
\label{eq:Bmu_bv}
\end{equation}
and, if Assumption~\ref{ass:cross_bernstein} holds,
\begin{align}
\maxnorm{\mathbf R_{\mu,n}}
=
O_p\!\left(
C_{d,1}m_n h_n\,V_{\mu,n}\,\frac{L_n\sqrt{\log p_n}}{n}
+
C_{d,1}M_\mu \frac{L_n\log p_n}{n}
\right).
\label{eq:Rmu_bv}
\end{align}
\end{proposition}

\begin{remark}[Plug-in rates for \(a_n\)]
\label{rem:an_plug_in}
Theorem~\ref{thm:maxnorm_rate} remains unchanged, but Propositions~\ref{prop:holder_mean_bounds} and \ref{prop:bv_mean_bounds} provide explicit, directly usable upper bounds for the mean-related components in \(a_n\). For example, under Assumption~\ref{ass:mean_piecewise_holder}, one may substitute \eqref{eq:Bmu_holder} and \eqref{eq:Rmu_holder} into \eqref{eq:an_def}.
\end{remark}

\begin{corollary}[Explicit \(a_n\) under piecewise-H\"older means]
\label{cor:an_holder_explicit}
Suppose the conditions of Theorem~\ref{thm:maxnorm_rate} hold, and Assumptions~\ref{ass:cross_bernstein} and \ref{ass:mean_piecewise_holder} hold. Then one may choose
\[
a_n
\asymp
\ell_n^{-q_0}
+\frac{(m_n+1)h_n}{n}
+\sqrt{\frac{L_n\log p_n}{n}}
+\frac{L_n\log p_n}{n}
+\ell_n\left\{
\left(\frac{h_n}{n}\right)^{2\beta} C_{d,\beta}^2
+
C_{d,1}^2 M_\mu^2 \frac{J_n m_n h_n}{n}
\right\}
\]
\[
\quad
+\sqrt{\left\{
\left(\frac{h_n}{n}\right)^{2\beta} C_{d,\beta}^2
+
C_{d,1}^2 M_\mu^2 \frac{J_n m_n h_n}{n}
\right\}}
L_n\sqrt{\frac{\log p_n}{n}}
+
\left\{
C_{d,\beta}\left(\frac{h_n}{n}\right)^\beta + C_{d,1}M_\mu
\right\}
\frac{L_n\log p_n}{n}.
\]
In particular, if the mean-induced terms are of smaller order than the oracle terms, the rate reduces to the bandwidth-balanced form in Corollary~\ref{cor:balanced_rate}.
\end{corollary}

\section{Simulation} \label{sec:simulation}

We conduct Monte Carlo experiments to assess finite-sample performance of high-dimensional long-run covariance (LRC) matrix estimation under temporal dependence and cross-sectional dependence. For each replication, we generate a $p$-variate time series $\{\bm{X}_t\}_{t=1}^n$ with different sample sizes $n=200,400,800,1600$ and dimensions $p=300,500$. The temporal dependence follows an AR(1) recursion
\[
\bm Z_t=\phi \bm Z_{t-1}+\bm{\varepsilon}_t,\qquad t=2,\ldots,n,
\]
with $\phi=0.5$, and we discard an initial burn-in of length 200 to mitigate initialization effects. Throughout, the innovations $\{\bm{\varepsilon}_t\}$ are i.i.d.\ Gaussian with mean zero and covariance matrix ${\bm\Sigma}_{\bm\varepsilon}$, whose structure differs across the following two models.

\paragraph{Model I (tridiagonal cross-sectional dependence).}
We impose local cross-sectional dependence by taking ${\bm\Sigma}_{\bm\varepsilon}$ to be tridiagonal. Specifically, let $a=0.5$ and define ${\bm\Sigma}_{\bm\varepsilon}$ by
\[
({\bm\Sigma}_{\bm\varepsilon})_{11}=1,\qquad ({\bm\Sigma}_{\bm\varepsilon})_{ii}=1+a^2~(i\ge 2),\qquad
({\bm\Sigma}_{\bm\varepsilon})_{i,i-1}=({\bm\Sigma}_{\bm\varepsilon})_{i-1,i}=a~(i\ge 2),
\]
and $({\bm\Sigma}_{\bm\varepsilon})_{ij}=0$ otherwise.

\paragraph{Model II (Toeplitz cross-sectional dependence).}
We consider a bandable Toeplitz covariance with exponential decay:
\[
({\bm\Sigma}_{\bm\varepsilon})_{ij}=\rho^{|i-j|},\qquad \rho=0.7,\qquad 1\le i,j\le p .
\]

\paragraph{Model III (Permuted Block Diagonal Structure).} 
We generate data from a multivariate model with covariance matrix
 \[
 \mathbf{\Sigma}=\mathbf{\Pi}\mathbf{\Sigma}_0\mathbf{\Pi}^\top,
 \qquad
 \mathbf{\Sigma}_0=\operatorname{diag}(\mathbf B,\ldots,\mathbf B),
 \qquad
 \mathbf B=
 \begin{pmatrix}
 1 & \rho\\
 \rho & 1
 \end{pmatrix},
 \]
where $\mathbf{\Pi}$ is a permutation matrix. This covariance matrix is highly sparse, since each row has only one nonzero off-diagonal entry, but it is not banded after permutation.

To examine robustness to nonconstant means, we add a deterministic mean component to the first $m=20$ coordinates of the observed series. Let $t=i/n$ for $i=1,\ldots,n$ and define
\[
\mu(t)=\exp(t)+\mathbf{1}(t>0.3)+2\mathbf{1}(t>0.6)+4\mathbf{1}(t>0.8).
\]
We set $\mathbb{E}(X_{t,j})=\mu(t)$ for $j=1,\ldots,m$ and $\mathbb{E}(X_{t,j})=0$ for $j>m$. The target LRC matrix is defined for the centered stationary component and, under both models, equals
\[
\mathbf{V} \;=\; \sum_{k=-\infty}^{\infty}\bm{\Gamma}(k)\;=\;\frac{{\bm\Sigma}_{\bm\varepsilon}}{(1-\phi)^2},
\]
where ${\bm\Gamma}(k)=\mathrm{Cov}(\bm{X}_t,\bm{X}_{t-k})$.

We compare six estimators of $\mathbf{V}$.

\begin{itemize}
\item[1.] \textbf{HAC.} The standard heteroskedasticity and autocorrelation consistent estimator of \citet{newey1987}. For a fair comparison, we use the same kernel as in the DB estimator, namely
\[
K_q(u)=(1-|u|^q)^+,\qquad q=2.
\]

\item[2.] \textbf{MAC.} The mean-structure and autocorrelation consistent covariance matrix estimator proposed by \citet{chan2022mac}.

\item[3.] \textbf{DB.} The difference-based estimator proposed by \citet{chan2022diff}. In our implementation, we set $m=3$ and use
\[
\bm d=(0.1942,\,0.2809,\,0.3832,\,-0.8582)^\top.
\]
We choose $h_n=2\ell_n$ and $\ell_n=\min\left\{\left\lfloor \left(\frac{n}{\log p}\right)^{1/4}\right\rfloor,\ \left\lfloor \frac{n-10}{28}\right\rfloor\right\}.
$

\item[4.] \textbf{Hard.} The hard-thresholded DB estimator defined in \eqref{eq:hard_threshold}.

\item[5.] \textbf{Soft.} The soft-thresholded DB estimator defined in \eqref{eq:soft_threshold}.

\item[6.] \textbf{Taper.} The tapered DB estimator defined in \eqref{eq:taper_est}, where the tapering weights are given by
\[
W^{(k)}_{ij}=
\begin{cases}
1, & |i-j|\le k/2,\\
2-2|i-j|/k, & k/2<|i-j|<k,\\
0, & |i-j|\ge k.
\end{cases}
\]
\end{itemize}

For each replication and each method, we compute the estimation error $\mathbf{E}=\widehat{\mathbf{V}}-\mathbf{V}$ and report four matrix norms:
\[
\|\mathbf{E}\|_{F},\qquad \|\mathbf{E}\|_{1},\qquad \|\mathbf{E}\|_{\max},\qquad \|\mathbf{E}\|_{2},
\]
where $\|\cdot\|_{F}$ is the Frobenius norm, $\|\cdot\|_{1}$ is the induced matrix $1$-norm, $\|\cdot\|_{\max}=\max_{i,j}|E_{ij}|$ is the entrywise maximum norm, and $\|\cdot\|_{2}$ is the spectral norm. We also report their relative errors, obtained by dividing each error norm by the corresponding norm of $\mathbf{V}$ (with a small numerical safeguard in the denominator). Results are summarized as Monte Carlo averages over 1000 replications.

Several conclusions emerge clearly from Tables~\ref{tab:lrc_errors}--\ref{tab:lrc_errors_model35}. First, the unregularized HAC estimator performs very poorly in all six settings. Its absolute and relative errors remain extremely large across all matrix norms, and the improvement with increasing sample size is negligible. This instability is especially pronounced under the induced $\ell_1$ and spectral norms, indicating that direct high-dimensional HAC estimation is not reliable in the present regime. In the settings where the MAC estimator is reported, it improves substantially over HAC, but it is still uniformly less accurate than the difference-based procedures, particularly after structural regularization. By contrast, the raw DB estimator already achieves a dramatic reduction in error, and its accuracy improves steadily as $n$ increases, confirming the practical advantage of differencing for removing the nonstationary mean component before long-run covariance estimation.

Second, the effect of regularization depends in an informative way on the underlying cross-sectional structure. For Models~I and II, where the covariance matrix is banded or bandable under the observed ordering, the regularized DB estimators uniformly improve upon the raw DB estimator, and the tapered estimator is typically the most accurate overall. Its advantage is especially clear under the Frobenius, matrix $\ell_1$, and spectral norms, and this pattern is stable for both $p=300$ and $p=500$. The gains from regularization are therefore not only substantial but also robust to increasing dimensionality. In contrast, Model~III is sparse but explicitly non-banded after permutation. In this setting, the tapering estimator loses much of its advantage, while hard- and soft-thresholding become more competitive and often outperform tapering under the Frobenius norm, with soft-thresholding also performing very well under the spectral norm in several cases. This contrast is consistent with the structural design of the three regularization schemes: tapering is most effective when the variable ordering carries meaningful locality information, whereas thresholding is better adapted to general sparsity without bandedness. Finally, across all six tables, the max-norm errors of the regularized estimators are often quite close to one another, suggesting that the main benefit of regularization lies in improving global matrix recovery rather than merely reducing the largest entrywise deviation.

\begin{table}[htbp]
\centering
\caption{Estimation errors of different methods under Model I with $p=300$.}
\label{tab:lrc_errors}
\resizebox{\textwidth}{!}{%
\begin{tabular}{lrrrrrrrr}
\toprule
 Method & $\|\bE\|_{F}$ & $\|\bE\|_{1}$ & $\|\bE\|_{\max}$ & $\|\bE\|_{2}$  & $\frac{\|\bE\|_{F}}{\|\bV\|_{F}}$ & $\frac{\|\bE\|_{1}}{\|\bV\|_{1}}$ & $\frac{\|\bE\|_{\max}}{\|\bV\|_{\max}}$ & $\frac{\|\bE\|_{2}}{\|\bV\|_{2}}$ \\
\midrule
\multicolumn{9}{c}{$(n,p)=(200,300)$}\\ \hline
 HAC   & 4106.84 & 6867.44 & 235.29 & 4088.23 & 41.31 & 763.05 & 47.06 & 454.26 \\
 MAC & 298.80 & 395.14 & 7.05 & 113.76 & 3.01 & 43.90 & 1.41 & 12.64 \\
 DB    &  114.23 &  130.85 &   3.47 &   33.43 &  1.15 &  14.54 &  0.69 &   3.71 \\
 Hard  &   61.08 &   21.90 &   3.47 &   10.77 &  0.61 &   2.43 &  0.69 &   1.20 \\
 Soft  &   60.16 &   16.83 &   3.47 &    7.38 &  0.61 &   1.87 &  0.69 &   0.82 \\
 Taper &   47.54 &    6.44 &   3.47 &    5.56 &  0.48 &   0.72 &  0.69 &   0.62 \\
\midrule
\multicolumn{9}{c}{$(n,p)=(400,300)$}\\ \hline
 HAC   & 4269.37 & 6320.69 & 237.59 & 4258.68 & 42.95 & 702.30 & 47.52 & 473.20 \\
 MAC & 233.92 & 280.08 & 4.94 & 74.85 & 2.35 & 31.12 & 0.99 & 8.32 \\
 DB    &   86.31 &   80.67 &   3.13 &   16.57 &  0.87 &   8.96 &  0.63 &   1.84 \\
 Hard  &   49.90 &    7.90 &   3.13 &    6.19 &  0.50 &   0.88 &  0.63 &   0.69 \\
 Soft  &   53.45 &   11.65 &   3.13 &    6.05 &  0.54 &   1.29 &  0.63 &   0.67 \\
 Taper &   45.44 &    5.83 &   3.13 &    5.07 &  0.46 &   0.65 &  0.63 &   0.56 \\
\midrule
\multicolumn{9}{c}{$(n,p)=(800,300)$}\\ \hline
 HAC   & 4350.53 & 5825.08 & 235.43 & 4344.89 & 43.76 & 647.23 & 47.09 & 482.78 \\
 MAC & 179.73 & 199.80 & 3.62 & 49.01 & 1.81 & 22.20 & 0.72 & 5.45 \\
 DB    &   83.31 &   82.84 &   2.45 &   17.53 &  0.84 &   9.20 &  0.49 &   1.95 \\
 Hard  &   33.81 &   10.17 &   2.45 &    4.96 &  0.34 &   1.13 &  0.49 &   0.55 \\
 Soft  &   41.38 &    6.51 &   2.45 &    5.05 &  0.42 &   0.72 &  0.49 &   0.56 \\
 Taper &   30.99 &    4.56 &   2.45 &    3.87 &  0.31 &   0.51 &  0.49 &   0.43 \\
\midrule
\multicolumn{9}{c}{$(n,p)=(1600,300)$}\\ \hline
 HAC   & 4378.63 & 5425.14 & 232.24 & 4375.73 & 44.04 & 602.79 & 46.45 & 486.20 \\
 MAC & 137.11 & 142.28 & 2.47 & 32.32 & 1.38 & 15.81 & 0.49 & 3.59 \\
 DB    &   72.96 &   71.94 &   1.89 &   15.02 &  0.73 &   7.99 &  0.38 &   1.67 \\
 Hard  &   27.09 &   11.41 &   1.89 &    5.21 &  0.27 &   1.27 &  0.38 &   0.58 \\
 Soft  &   30.45 &    8.61 &   1.89 &    4.01 &  0.31 &   0.96 &  0.38 &   0.45 \\
 Taper &   21.98 &    3.57 &   1.89 &    2.97 &  0.22 &   0.40 &  0.38 &   0.33 \\
\bottomrule
\end{tabular}%
}
\end{table}

\begin{table}[htbp]
\centering
\caption{Estimation errors of different methods under Model II with $p=300$.}
\label{tab:lrc_errors_model2}
\resizebox{\textwidth}{!}{%
\begin{tabular}{lrrrrrrrr}
\toprule
Method & $\|\bE\|_{F}$ & $\|\bE\|_{1}$ & $\|\bE\|_{\max}$ & $\|\bE\|_{2}$  & $\frac{\|\bE\|_{F}}{\|\bV\|_{F}}$ & $\frac{\|\bE\|_{1}}{\|\bV\|_{1}}$ & $\frac{\|\bE\|_{\max}}{\|\bV\|_{\max}}$ & $\frac{\|\bE\|_{2}}{\|\bV\|_{2}}$ \\
\midrule
\multicolumn{9}{c}{$(n,p)=(200,300)$}\\ \hline
HAC   & 4054.94 & 6474.14 & 226.25 & 4043.14 & 34.32 & 285.62 & 56.56 & 178.52 \\
MAC & 241.73 & 334.23 & 5.88 & 106.17 & 2.05 & 14.75 & 1.47 & 4.69 \\
DB    &  100.44 &  108.19 &   2.78 &   28.82 &  0.85 &   4.77 &  0.69 &   1.27 \\
Hard  &   75.79 &   26.32 &   2.83 &   17.90 &  0.64 &   1.16 &  0.71 &   0.79 \\
Soft  &   78.71 &   27.24 &   2.78 &   17.91 &  0.67 &   1.20 &  0.70 &   0.79 \\
Taper &   59.71 &   17.73 &   2.78 &   14.31 &  0.51 &   0.78 &  0.69 &   0.63 \\
\midrule
\multicolumn{9}{c}{$(n,p)=(400,300)$}\\ \hline
HAC   & 4248.96 & 6027.40 & 230.56 & 4242.28 & 35.96 & 265.91 & 57.64 & 187.31 \\
MAC & 189.48 & 237.01 & 4.07 & 72.58 & 1.60 & 10.46 & 1.02 & 3.20 \\
DB    &   79.35 &   67.94 &   2.51 &   17.12 &  0.67 &   3.00 &  0.63 &   0.76 \\
Hard  &   64.03 &   18.88 &   2.51 &   15.40 &  0.54 &   0.83 &  0.63 &   0.68 \\
Soft  &   69.15 &   22.26 &   2.51 &   16.01 &  0.59 &   0.98 &  0.63 &   0.71 \\
Taper &   55.53 &   15.97 &   2.51 &   12.93 &  0.47 &   0.70 &  0.63 &   0.57 \\
\midrule
\multicolumn{9}{c}{$(n,p)=(800,300)$}\\ \hline
HAC   & 4341.36 & 5634.79 & 230.95 & 4337.80 & 36.74 & 248.59 & 57.74 & 191.53 \\
MAC & 146.19 & 170.41 & 2.90 & 48.77 & 1.24 & 7.52 & 0.73 & 2.15 \\
DB    &   71.74 &   69.05 &   1.95 &   16.10 &  0.61 &   3.05 &  0.49 &   0.71 \\
Hard  &   46.80 &   16.35 &   1.96 &   11.97 &  0.40 &   0.72 &  0.49 &   0.53 \\
Soft  &   56.43 &   17.77 &   1.97 &   14.03 &  0.48 &   0.78 &  0.49 &   0.62 \\
Taper &   39.26 &   13.12 &   1.95 &   10.04 &  0.33 &   0.58 &  0.49 &   0.44 \\
\midrule
\multicolumn{9}{c}{$(n,p)=(1600,300)$}\\ \hline
HAC   & 4372.59 & 5291.68 & 229.19 & 4370.74 & 37.00 & 233.46 & 57.30 & 192.99 \\
MAC & 112.72 & 123.15 & 2.01 & 33.86 & 0.95 & 5.43 & 0.50 & 1.50 \\
DB    &   61.33 &   60.52 &   1.51 &   14.43 &  0.52 &   2.67 &  0.38 &   0.64 \\
Hard  &   35.56 &   15.01 &   1.51 &    9.27 &  0.30 &   0.66 &  0.38 &   0.41 \\
Soft  &   42.28 &   17.35 &   1.52 &   11.02 &  0.36 &   0.77 &  0.38 &   0.49 \\
Taper &   28.73 &   10.70 &   1.51 &    7.69 &  0.24 &   0.47 &  0.38 &   0.34 \\
\bottomrule
\end{tabular}%
}
\end{table}

\begin{table}[htbp]
\centering
\caption{Estimation errors of different methods under Model III with $p=300$.}
\label{tab:lrc_errors_model33}
\resizebox{\textwidth}{!}{%
\begin{tabular}{lrrrrrrrr}
\toprule
 Method & $\|\bE\|_{F}$ & $\|\bE\|_{1}$ & $\|\bE\|_{\max}$ & $\|\bE\|_{2}$ &  $\frac{\|\bE\|_{F}}{\|\bV\|_{F}}$ & $\frac{\|\bE\|_{1}}{\|\bV\|_{1}}$ & $\frac{\|\bE\|_{\max}}{\|\bV\|_{\max}}$ & $\frac{\|\bE\|_{2}}{\|\bV\|_{2}}$ \\
\midrule
\multicolumn{9}{c}{$(n,p)=(200,300)$}\\ \hline
 HAC   & 4073.82 & 6571.00 & 233.42 & 4061.75 & 52.59 & 1095.17 & 58.35 & 676.96 \\
 DB    &   92.37 &  111.94 &   2.78 &   30.07 &  1.19 &   18.66 &  0.70 &   5.01 \\
 Hard  &   46.23 &   26.73 &   2.78 &   14.74 &  0.60 &    4.45 &  0.70 &   2.46 \\
 Soft  &   45.46 &   16.64 &   2.78 &    7.44 &  0.59 &    2.77 &  0.70 &   1.24 \\
 Taper &   47.57 &    4.78 &   2.78 &    4.57 &  0.61 &    0.80 &  0.70 &   0.76 \\
\midrule
\multicolumn{9}{c}{$(n,p)=(400,300)$}\\ \hline
 HAC   & 4253.51 & 6094.31 & 235.23 & 4246.63 & 54.91 & 1015.72 & 58.81 & 707.77 \\
 DB    &   69.03 &   67.77 &   2.52 &   14.60 &  0.89 &   11.30 &  0.63 &   2.43 \\
 Hard  &   36.64 &    8.79 &   2.52 &    4.41 &  0.47 &    1.47 &  0.63 &   0.73 \\
 Soft  &   40.29 &   10.47 &   2.52 &    4.28 &  0.52 &    1.75 &  0.63 &   0.71 \\
 Taper &   46.80 &    4.52 &   2.52 &    4.34 &  0.60 &    0.75 &  0.63 &   0.72 \\
\midrule
\multicolumn{9}{c}{$(n,p)=(800,300)$}\\ \hline
 HAC   & 4338.73 & 5657.96 & 232.83 & 4335.09 & 56.01 &  942.99 & 58.21 & 722.52 \\
 DB    &   67.04 &   69.93 &   1.96 &   15.66 &  0.87 &   11.65 &  0.49 &   2.61 \\
 Hard  &   27.07 &   12.86 &   1.96 &    6.23 &  0.35 &    2.14 &  0.49 &   1.04 \\
 Soft  &   30.29 &    5.59 &   1.96 &    3.50 &  0.39 &    0.93 &  0.49 &   0.58 \\
 Taper &   40.61 &    3.96 &   2.02 &    3.78 &  0.52 &    0.66 &  0.50 &   0.63 \\
\midrule
\multicolumn{9}{c}{$(n,p)=(1600,300)$}\\ \hline
 HAC   & 4374.31 & 5317.32 & 231.56 & 4372.45 & 56.47 &  886.22 & 57.89 & 728.74 \\
 DB    &   58.78 &   61.05 &   1.52 &   13.48 &  0.76 &   10.17 &  0.38 &   2.25 \\
 Hard  &   22.23 &   12.88 &   1.52 &    6.67 &  0.29 &    2.15 &  0.38 &   1.11 \\
 Soft  &   22.29 &    8.30 &   1.52 &    3.61 &  0.29 &    1.38 &  0.38 &   0.60 \\
 Taper &   37.78 &    3.84 &   2.00 &    3.42 &  0.49 &    0.64 &  0.50 &   0.57 \\
\bottomrule
\end{tabular}%
}
\end{table}

\begin{table}[htbp]
\centering
\caption{Estimation errors of different methods under Model I with $p=500$.}
\label{tab:lrc_errors_four_runs}
\resizebox{\textwidth}{!}{%
\begin{tabular}{lrrrrrrrr}
\toprule
Method & $\|\bE\|_{F}$ & $\|\bE\|_{1}$ & $\|\bE\|_{\max}$ & $\|\bE\|_{2}$
& $\frac{\|\bE\|_{F}}{\|\bV\|_{F}}$ & $\frac{\|\bE\|_{1}}{\|\bV\|_{1}}$ & $\frac{\|\bE\|_{\max}}{\|\bV\|_{\max}}$ & $\frac{\|\bE\|_{2}}{\|\bV\|_{2}}$ \\
\midrule
\multicolumn{9}{c}{$(n,p)=(200,500)$}\\ \hline
HAC   & 4257.80 & 8742.62 & 235.88 & 4207.16 & 33.16 & 971.40 & 47.18 & 467.47 \\
MAC & 489.54 & 627.42 & 6.85 & 163.23 & 3.81 & 69.71 & 1.37 & 18.14 \\
DB    &  181.36 &  204.51 &   3.53 &   47.34 &  1.41 &  22.72 &  0.71 &   5.26 \\
Hard  &   79.14 &   21.24 &   3.53 &   10.08 &  0.62 &   2.36 &  0.71 &   1.12 \\
Soft  &   78.73 &   21.49 &   3.53 &    7.82 &  0.61 &   2.39 &  0.71 &   0.87 \\
Taper &   61.79 &    6.61 &   3.53 &    5.66 &  0.48 &   0.73 &  0.71 &   0.63 \\
\midrule
\multicolumn{9}{c}{$(n,p)=(400,500)$}\\ \hline
HAC   & 4345.56 & 7651.68 & 238.36 & 4316.10 & 33.85 & 850.19 & 47.67 & 479.57 \\
MAC & 386.65 & 452.15 & 5.16 & 107.04 & 3.01 & 50.24 & 1.03 & 11.89 \\
DB    &  135.03 &  130.78 &   3.17 &   24.70 &  1.05 &  14.53 &  0.63 &   2.74 \\
Hard  &   65.93 &    7.82 &   3.17 &    6.34 &  0.51 &   0.87 &  0.63 &   0.70 \\
Soft  &   70.54 &   14.23 &   3.17 &    6.22 &  0.55 &   1.58 &  0.63 &   0.69 \\
Taper &   59.38 &    6.36 &   3.17 &    5.16 &  0.46 &   0.71 &  0.63 &   0.57 \\
\midrule
\multicolumn{9}{c}{$(n,p)=(800,500)$}\\ \hline
HAC   & 4394.94 & 6790.50 & 237.17 & 4379.21 & 34.23 & 754.50 & 47.43 & 486.58 \\
MAC & 298.51 & 323.50 & 3.83 & 69.46 & 2.33 & 35.94 & 0.77 & 7.72 \\
DB    &  134.28 &  131.98 &   2.50 &   24.28 &  1.05 &  14.66 &  0.50 &   2.70 \\
Hard  &   43.87 &   10.33 &   2.50 &    5.10 &  0.34 &   1.15 &  0.50 &   0.57 \\
Soft  &   54.42 &    7.53 &   2.50 &    5.14 &  0.42 &   0.84 &  0.50 &   0.57 \\
Taper &   40.32 &    4.67 &   2.50 &    3.94 &  0.31 &   0.52 &  0.50 &   0.44 \\
\midrule
\multicolumn{9}{c}{$(n,p)=(1600,500)$}\\ \hline
HAC   & 4400.11 & 6097.05 & 232.43 & 4392.01 & 34.27 & 677.45 & 46.49 & 488.01 \\
MAC & 226.50 & 230.85 & 2.77 & 45.52 & 1.76 & 25.65 & 0.55 & 5.06 \\
DB    &  118.73 &  114.32 &   1.94 &   19.73 &  0.92 &  12.70 &  0.39 &   2.19 \\
Hard  &   35.35 &   10.90 &   1.95 &    4.73 &  0.28 &   1.21 &  0.39 &   0.53 \\
Soft  &   41.13 &   11.08 &   1.94 &    4.19 &  0.32 &   1.23 &  0.39 &   0.47 \\
Taper &   28.59 &    3.70 &   1.94 &    3.07 &  0.22 &   0.41 &  0.39 &   0.34 \\
\bottomrule
\end{tabular}}
\end{table}

\begin{table}[htbp]
\centering
\caption{Estimation errors of different methods under Model II with $p=500$.}
\label{tab:lrc_errors_four_runs_new}
\resizebox{\textwidth}{!}{%
\begin{tabular}{lrrrrrrrr}
\toprule
Method & $\|\bE\|_{F}$ & $\|\bE\|_{1}$ & $\|\bE\|_{\max}$ & $\|\bE\|_{2}$
& $\frac{\|\bE\|_{F}}{\|\bV\|_{F}}$ & $\frac{\|\bE\|_{1}}{\|\bV\|_{1}}$ & $\frac{\|\bE\|_{\max}}{\|\bV\|_{\max}}$ & $\frac{\|\bE\|_{2}}{\|\bV\|_{2}}$ \\
\midrule
\multicolumn{9}{c}{$(n,p)=(200,500)$}\\ \hline
HAC   & 4179.92 & 8117.72 & 227.08 & 4147.18 & 27.38 & 358.13 & 56.77 & 183.02 \\
MAC   & 393.19 & 527.70 & 5.87 & 148.61 & 2.58 & 23.28 & 1.47 & 6.56 \\
DB    & 154.85 & 168.22 & 2.82 & 40.67 & 1.01 & 7.42 & 0.71 & 1.79 \\
Hard  & 99.55 & 26.01 & 2.84 & 18.40 & 0.65 & 1.15 & 0.71 & 0.81 \\
Soft  & 104.16 & 30.70 & 2.82 & 18.36 & 0.68 & 1.35 & 0.71 & 0.81 \\
Taper & 77.20 & 17.95 & 2.82 & 14.45 & 0.51 & 0.79 & 0.71 & 0.64 \\
\midrule
\multicolumn{9}{c}{$(n,p)=(400,500)$}\\ \hline
HAC   & 4305.60 & 7216.19 & 232.09 & 4286.98 & 28.20 & 318.36 & 58.02 & 189.19 \\
MAC   & 309.61 & 372.73 & 4.09 & 101.23 & 2.03 & 16.44 & 1.02 & 4.47 \\
DB    & 119.57 & 108.33 & 2.53 & 21.06 & 0.78 & 4.78 & 0.63 & 0.93 \\
Hard  & 84.39 & 19.18 & 2.54 & 15.74 & 0.55 & 0.85 & 0.63 & 0.69 \\
Soft  & 91.95 & 24.61 & 2.53 & 16.43 & 0.60 & 1.09 & 0.63 & 0.73 \\
Taper & 72.51 & 16.64 & 2.53 & 13.11 & 0.47 & 0.73 & 0.63 & 0.58 \\
\midrule
\multicolumn{9}{c}{$(n,p)=(800,500)$}\\ \hline
HAC   & 4383.82 & 6491.83 & 232.63 & 4373.89 & 28.71 & 286.40 & 58.16 & 193.02 \\
MAC   & 239.79 & 270.08 & 3.03 & 68.07 & 1.57 & 11.92 & 0.76 & 3.00 \\
DB    & 112.91 & 109.69 & 1.99 & 22.36 & 0.74 & 4.84 & 0.50 & 0.99 \\
Hard  & 61.36 & 16.86 & 1.99 & 12.21 & 0.40 & 0.74 & 0.50 & 0.54 \\
Soft  & 75.20 & 18.73 & 2.02 & 14.40 & 0.49 & 0.83 & 0.50 & 0.64 \\
Taper & 50.91 & 13.34 & 1.99 & 10.13 & 0.33 & 0.59 & 0.50 & 0.45 \\
\midrule
\multicolumn{9}{c}{$(n,p)=(1600,500)$}\\ \hline
HAC   & 4386.06 & 5881.84 & 229.14 & 4380.89 & 28.73 & 259.49 & 57.28 & 193.33 \\
MAC   & 183.70 & 194.87 & 2.19 & 46.32 & 1.20 & 8.60 & 0.55 & 2.04 \\
DB    & 98.07 & 94.97 & 1.55 & 19.01 & 0.64 & 4.19 & 0.39 & 0.84 \\
Hard  & 46.83 & 15.11 & 1.55 & 9.46 & 0.31 & 0.67 & 0.39 & 0.42 \\
Soft  & 56.96 & 19.62 & 1.57 & 11.48 & 0.37 & 0.87 & 0.39 & 0.51 \\
Taper & 37.19 & 11.01 & 1.55 & 7.87 & 0.24 & 0.49 & 0.39 & 0.35 \\
\bottomrule
\end{tabular}}
\end{table}

\begin{table}[htbp]
\centering
\caption{Estimation errors of different methods under Model III with $p=500$.}
\label{tab:lrc_errors_model35}
\resizebox{\textwidth}{!}{%
\begin{tabular}{lrrrrrrrr}
\toprule
 Method & $\|\bE\|_{F}$ & $\|\bE\|_{1}$ & $\|\bE\|_{\max}$ & $\|\bE\|_{2}$ & rel\_frob & rel\_$\ell_1$ & rel\_$\ell_\infty$ & rel\_spec \\
\midrule
\multicolumn{9}{c}{$(n,p)=(200,300)$}\\ \hline
 HAC   & 4187.35 & 8199.56 & 232.71 & 4154.57 & 41.87 & 1366.59 & 58.18 & 692.43 \\
 DB    &  145.68 &  174.69 &   2.82 &   39.86 &  1.46 &   29.12 &  0.71 &   6.64 \\
 Hard  &   58.45 &   26.97 &   2.82 &   13.98 &  0.58 &    4.50 &  0.71 &   2.33 \\
 Soft  &   59.05 &   21.96 &   2.82 &    8.28 &  0.59 &    3.66 &  0.71 &   1.38 \\
 Taper &   61.66 &    4.82 &   2.82 &    4.61 &  0.62 &    0.80 &  0.71 &   0.77 \\
\midrule
\multicolumn{9}{c}{$(n,p)=(400,300)$}\\ \hline
 HAC   & 4307.64 & 7275.95 & 235.89 & 4288.61 & 43.08 & 1212.66 & 58.97 & 714.77 \\
 DB    &  107.91 &  107.68 &   2.53 &   20.23 &  1.08 &   17.95 &  0.63 &   3.37 \\
 Hard  &   47.30 &    8.23 &   2.53 &    4.44 &  0.47 &    1.37 &  0.63 &   0.74 \\
 Soft  &   52.74 &   12.30 &   2.53 &    4.44 &  0.53 &    2.05 &  0.63 &   0.74 \\
 Taper &   60.89 &    5.37 &   2.53 &    4.47 &  0.61 &    0.90 &  0.63 &   0.75 \\
\midrule
\multicolumn{9}{c}{$(n,p)=(800,300)$}\\ \hline
 HAC   & 4371.29 & 6505.31 & 234.85 & 4361.16 & 43.71 & 1084.22 & 58.71 & 726.86 \\
 DB    &  107.67 &  109.62 &   2.01 &   20.34 &  1.08 &   18.27 &  0.50 &   3.39 \\
 Hard  &   34.32 &   12.76 &   2.01 &    5.96 &  0.34 &    2.13 &  0.50 &   0.99 \\
 Soft  &   39.36 &    6.48 &   2.01 &    3.65 &  0.39 &    1.08 &  0.50 &   0.61 \\
 Taper &   52.54 &    4.01 &   2.04 &    3.81 &  0.53 &    0.67 &  0.51 &   0.63 \\
\midrule
\multicolumn{9}{c}{$(n,p)=(1600,300)$}\\ \hline
 HAC   & 4390.88 & 5917.29 & 231.45 & 4385.66 & 43.91 &  986.22 & 57.86 & 730.94 \\
 DB    &   95.27 &   94.72 &   1.55 &   16.52 &  0.95 &   15.79 &  0.39 &   2.75 \\
 Hard  &   28.36 &   12.80 &   1.55 &    6.20 &  0.28 &    2.13 &  0.39 &   1.03 \\
 Soft  &   29.95 &   10.70 &   1.55 &    3.81 &  0.30 &    1.78 &  0.39 &   0.64 \\
 Taper &   48.92 &    4.00 &   2.00 &    3.45 &  0.49 &    0.67 &  0.50 &   0.58 \\
\bottomrule
\end{tabular}%
}
\end{table}

\section{Data Application}\label{sec:data}
Following \citet{WangLiuFeng2025}, we study weekly log-returns of NASDAQ constituent stocks from January 2016 to December 2024. Financial return panels of this type are well known to exhibit substantial temporal dependence and strong cross-sectional comovement, so a realistic changepoint analysis must account for the long-run covariance structure; see also the empirical discussion in \citet{WangLiuFeng2025}. Our objective is to examine whether the high-dimensional return process experienced a structural break during the sample period.

Formally, let $\bm{X}_1,\ldots,\bm{X}_n\in\mathbb{R}^p$ denote the weekly return vectors where in our processed sample $n=471$ and $p=1556$. We consider the single-changepoint testing problem
\[
H_0:\ {\bm\mu}_1=\cdots={\bm\mu}_n
\qquad\text{versus}\qquad
H_1:\ {\bm\mu}_1=\cdots={\bm\mu}_{\tau}\neq {\bm\mu}_{\tau+1}=\cdots={\bm\mu}_n
\]
for some unknown $\tau\in\{1,\ldots,n-1\}$. In the present section, we focus exclusively on the dense-alternative statistic $\bm{S}_{n,p}$ proposed by \citet{WangLiuFeng2025}. Specifically, for each candidate split point $k$, let $\bm{W}(k)$ denote the quadratic CUSUM-type statistic and let $\widehat{\bm\mu}_{M,k}$ be the corresponding centering term defined in that paper. We then analyze the normalized process
\[
\widetilde{\bm{S}}_{n,p}(k)=\frac{\bm{W}(k)-\widehat{\bm\mu}_{M,k}}{\widehat\omega},
\qquad 1\le k\le n-1,
\]
and estimate the changepoint by the maximizer of $\widetilde{\bm{S}}_{n,p}(k)$ over $k$.

The only modification relative to \citet{WangLiuFeng2025} is in the estimation of the scaling factor $\omega$. Instead of using the original plug-in estimator based on the long-run covariance structure appearing in that paper, we replace $\mathbf{V}$ in the quantity $\{2\,\mathrm{tr}({\mathbf{V}}^2)/p\}^{1/2}$ by our sparse difference-based long-run covariance estimator. This yields three versions,
\[
\widehat\omega_{\mathrm{hard}}
=
\biggl\{\frac{2}{p}\,\mathrm{tr}\bigl(\widehat{\mathbf{V}}_{\mathrm{hard}}^2\bigr)\biggr\}^{1/2},
\qquad
\widehat\omega_{\mathrm{soft}}
=
\biggl\{\frac{2}{p}\,\mathrm{tr}\bigl(\widehat{\mathbf{V}}_{\mathrm{soft}}^2\bigr)\biggr\}^{1/2},
\qquad
\widehat\omega_{\mathrm{taper}}
=
\biggl\{\frac{2}{p}\,\mathrm{tr}\bigl(\widehat{\mathbf{V}}_{\mathrm{taper}}^2\bigr)\biggr\}^{1/2},
\]
which in turn produce three normalized $\bm{S}_{n,p}$ paths. The corresponding changepoint estimators are defined by
\[
\widehat\tau_{\mathrm{hard}}=\arg\max_{1\le k\le n-1}\widetilde{\bm{S}}^{(\mathrm{hard})}_{n,p}(k),
\qquad
\widehat\tau_{\mathrm{soft}}=\arg\max_{1\le k\le n-1}\widetilde{\bm{S}}^{(\mathrm{soft})}_{n,p}(k),
\qquad
\widehat\tau_{\mathrm{taper}}=\arg\max_{1\le k\le n-1}\widetilde{\bm{S}}^{(\mathrm{taper})}_{n,p}(k).
\]

\begin{figure}[htb]
	\centering
	\includegraphics[width=\linewidth]{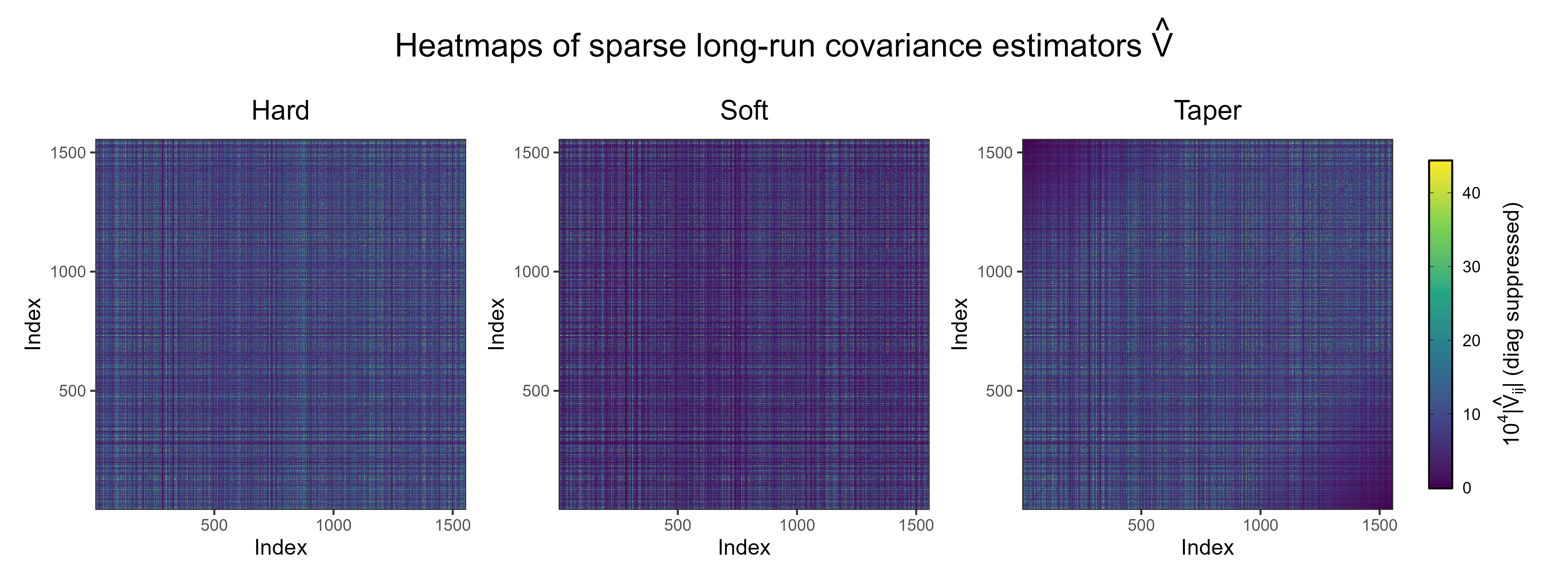}
 \caption{Heatmaps of sparse long run covariance matrix estimators (all entries are multiplied by 10000).}\label{fig:V_heatmap}
\end{figure}

Figure~\ref{fig:V_heatmap} provides additional empirical support for the sparse long-run covariance formulation adopted in our modification of the $\bm{S}_{n,p}$ statistic. Across the hard-thresholding, soft-thresholding, and tapering estimators, the estimated long-run covariance matrix $\widehat{\mathbf{V}}$ exhibits a pronounced near-sparse pattern: most off-diagonal entries are heavily shrunk toward zero, while only a relatively small fraction of entries remain visibly non-negligible. This feature is particularly evident away from the main diagonal, where large connected blocks are largely absent and dependence appears to be concentrated in a limited subset of local coordinates. Such a pattern is consistent with the working assumption that the cross-sectional and serially aggregated dependence structure is globally weak but locally persistent, thereby making sparse regularization a natural device for stabilizing the estimation of $\mathbf{V}$ in high dimensions.

The three regularization schemes reveal a broadly similar structural picture, although they differ in the degree of shrinkage. The hard-thresholding estimator preserves a slightly larger number of moderate signals, the soft-thresholding estimator induces more aggressive global shrinkage, and the tapering estimator retains a smoother decay pattern. Importantly, however, all three estimators lead to the same qualitative conclusion that $\widehat{\mathbf{V}}$ is approximately sparse rather than dense. This visual evidence complements the change-point analysis based on the modified $\bm{S}_{n,p}$ statistic, and supports our use of the sparse estimator $\widehat{\mathbf{V}}$ as a plug-in device in the normalization term involving $\tr({\mathbf{V}}^2)$.

Figure~\ref{fig:nasdaq_sparse_snp} displays the three normalized trajectories based on the hard-thresholding, soft-thresholding, and tapering estimators of the long-run covariance matrix. Although the three normalizations produce different peak magnitudes, they all attain their maximum at the same location, namely $\widehat{k}=269$, corresponding to 2021-02-19. As reported in Table~\ref{tab:nasdaq_sparse_snp}, the soft-thresholding version yields the largest peak value, $\max_k \widetilde{\bm{S}}_{n,p}(k)=0.577$, together with the smallest Monte Carlo $p$-value of $0.078$, while the hard-thresholding and tapering versions give weaker signals, with corresponding $p$-values $0.173$ and $0.134$, respectively. This pattern is consistent with the behavior observed under Model III in the simulation study, where soft thresholding tends to outperform hard thresholding in settings with less sharply structured sparsity.

From a financial perspective, the relatively weaker performance of tapering is not surprising. Tapering is most effective when the index ordering carries structural information, so that entries farther away from the diagonal are naturally less important. In the present application, however, the ordering of stocks in the data matrix is essentially nominal, and there is no intrinsic reason why stocks with nearby indices should exhibit stronger dependence than those far apart. Consequently, shrinking the estimated long-run covariance matrix according to index distance may be less appropriate in this setting. Overall, despite these differences in signal strength, all three sparse estimators identify exactly the same changepoint date, which provides strong evidence that the estimated break location is stable with respect to the specific sparse regularization used in the normalization step.

From an economic perspective, the date 2021-02-19 is closely aligned with the market regime shift discussed in \citet{WangLiuFeng2025}. Their interpretation links the February 2021 break to a broad repricing of technology and other long-duration growth stocks during the post-pandemic reopening phase, when abundant liquidity began to be offset by a rapid rise in U.S. Treasury yields and a corresponding increase in discount-rate pressure on growth valuations. This mechanism is particularly plausible for NASDAQ data, because the index is heavily tilted toward large technology and growth firms whose valuations are especially sensitive to changes in long-horizon discount rates. In statistical terms, such an episode is consistent with a relatively dense cross-sectional shift, which is precisely the type of alternative for which the $\bm{S}_{n,p}$ statistic is designed.

\begin{table}[htbp]
\centering
\caption{Modified $\bm{S}_{n,p}$ analysis based on sparse long-run covariance normalization.}
\label{tab:nasdaq_sparse_snp}
\begin{tabular}{lcccccc}
\toprule
Method & Tuning & $\widehat{\omega}$ & $\max_k \widetilde{\bm{S}}_{n,p}(k)$ & $\widehat{k}$ & Date &  $p$-value \\
\midrule
Hard  & $7.27\times 10^{-6}$ & 0.0724 & 0.484 & 269 & 2021-02-19 & 0.173 \\
Soft  & $3.07\times 10^{-4}$ & 0.0607 & 0.577 & 269 & 2021-02-19 & 0.078 \\
Taper & 1556                 & 0.0678 & 0.517 & 269 & 2021-02-19 & 0.134 \\
\bottomrule
\end{tabular}
\end{table}

\begin{figure}[htbp]
\centering
\includegraphics[width=\textwidth]{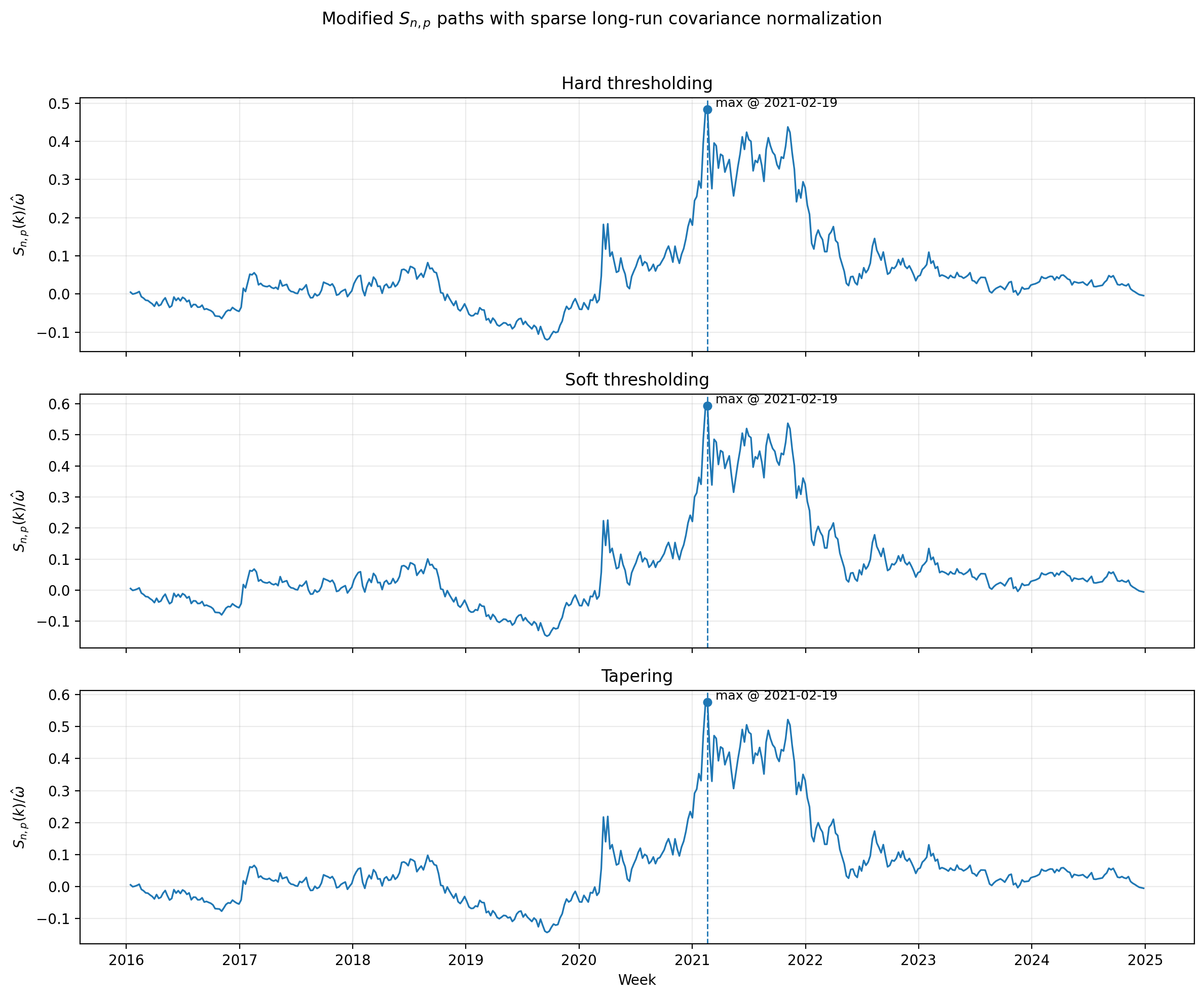}
\caption{Normalized trajectories of the modified $\bm{S}_{n,p}$ statistic based on hard-thresholding, soft-thresholding, and tapering estimators of the sparse long-run covariance matrix. The dashed vertical line marks the common maximizer of the three curves.}
\label{fig:nasdaq_sparse_snp}
\end{figure}

\section{Conclusion}\label{sec:conclusion}
This paper studies difference-based estimation of the high-dimensional long-run covariance matrix for temporally dependent time series in the presence of nonstationary mean structures. We show that the proposed difference-based estimator can effectively remove deterministic mean contamination and provide a reliable starting point for high-dimensional long-run covariance estimation. Building on this estimator, we further investigate hard-thresholding, soft-thresholding, and tapering regularization methods, and establish their theoretical convergence properties under suitable sparsity and dependence conditions. The numerical results demonstrate that the regularized difference-based procedures substantially improve estimation accuracy over classical HAC-type methods and remain effective across a range of dependence structures and sample sizes.

There are several important directions for future research. First, the estimated high-dimensional long-run covariance matrix can be incorporated into other statistical inference problems involving dependent data, such as linear and quadratic discriminant analysis, as well as related classification and testing problems; see, for example, \citet{Friedman1989RDA,BickelLevina2004,FanFengTong2012}. Second, the current framework is mainly developed for sparse long-run covariance structures. It would be of substantial interest to extend the methodology to nonsparse settings by combining difference-based ideas with low-rank or approximate factor structures, thereby accommodating more complex dependence patterns commonly encountered in finance, economics, and other high-dimensional applications; see \citet{FanFanLv2008,ChamberlainRothschild1983,FanLiaoMincheva2013}.

\appendix
\section{Appendix: Proofs of Theorems}
\addcontentsline{toc}{section}{Appendix: Proofs}

\subsection{Proof of Theorem \ref{thm:maxnorm_rate}}
By \eqref{eq:Rmu},
\[
\Vdb - \Vtrue
=
(\Vor-\Vtrue) + \mathbf B_{\mu,n} + \mathbf R_{\mu,n}.
\]
Hence,
\begin{equation}
\maxnorm{\Vdb-\Vtrue}
\le
\maxnorm{\Vor-\Vtrue}
+\maxnorm{\mathbf B_{\mu,n}}
+\maxnorm{\mathbf R_{\mu,n}}.
\label{eq:proof_t1_decomp0}
\end{equation}
By Assumption~\ref{ass:mean_remainder},
\[
\maxnorm{\mathbf R_{\mu,n}}=O_p(r_{\mu,n}).
\]
Therefore, it remains to bound \(\maxnorm{\Vor-\Vtrue}\).
Write
\[
\Vor-\Vtrue
=
\underbrace{\bigl(\Vor-\E\Vor\bigr)}_{\text{oracle fluctuation}}
+
\underbrace{\bigl(\E\Vor-\Vtrue\bigr)}_{\text{oracle bias}}.
\]
Taking max norm yields
\begin{equation}
\maxnorm{\Vor-\Vtrue}
\le
\maxnorm{\Vor-\E\Vor}
+
\maxnorm{\E\Vor-\Vtrue}.
\label{eq:proof_t1_decomp1}
\end{equation}
By Assumption~\ref{ass:oracle_bias},
\begin{equation}
\maxnorm{\E\Vor-\Vtrue}
\le
C_B\left\{\ell_n^{-q_0}+\frac{(m_n+1)h_n}{n}\right\}.
\label{eq:proof_t1_bias}
\end{equation}

Thus the key step is to derive a high-dimensional max-norm concentration bound for \(\Vor-\E\Vor\). Define
\[
S_{rs}:=\Vor[r,s]-\E\Vor[r,s],\qquad 1\le r,s\le p_n.
\]
For any \(t>0\),
\[
\Prob\!\left(\max_{1\le r,s\le p_n}|S_{rs}|>t\right)
\le
\sum_{r=1}^{p_n}\sum_{s=1}^{p_n}\Prob(|S_{rs}|>t).
\]
Applying Assumption~\ref{ass:bernstein},
\begin{equation}
\Prob\!\left(\max_{r,s}|S_{rs}|>t\right)
\le
2p_n^2 \exp\!\left[
-c_1 n \min\!\left\{
\frac{t^2}{\ell_n+m_n h_n},\,
\frac{t}{\ell_n+m_n h_n}
\right\}
\right].
\label{eq:proof_t1_union}
\end{equation}

We now choose \(t\) to balance the polynomial factor \(p_n^2\).
Set
\[
t_1 := M_1 \sqrt{\frac{(\ell_n+m_n h_n)\log p_n}{n}}
\]
for a constant \(M_1>0\) to be chosen. Then
\[
\frac{n t_1^2}{\ell_n+m_n h_n}=M_1^2 \log p_n.
\]
Substituting into \eqref{eq:proof_t1_union} gives
\[
\Prob\!\left(\max_{r,s}|S_{rs}|>t_1\right)
\le
2p_n^{2-c_1M_1^2}
\to 0
\]
as soon as \(M_1\) is chosen large enough so that \(c_1M_1^2>3\), say. Therefore,
\begin{equation*}
\maxnorm{\Vor-\E\Vor}
=
O_p\!\left(
\sqrt{\frac{(\ell_n+m_n h_n)\log p_n}{n}}
\right)
\end{equation*}
whenever the quadratic term is the active branch.

To obtain a uniform bound valid across the two branches of the Bernstein inequality, define
\[
t_2 := M_2 \frac{(\ell_n+m_n h_n)\log p_n}{n},
\]
with \(M_2>0\) large. Then
\[
\frac{n t_2}{\ell_n+m_n h_n}=M_2\log p_n,
\]
hence
\[
\Prob\!\left(\max_{r,s}|S_{rs}|>t_2\right)
\le
2p_n^{2-c_1M_2}\to 0
\]
if \(M_2\) is sufficiently large. Combining the two regimes yields the standard Bernstein-type max-norm bound
\begin{equation}
\maxnorm{\Vor-\E\Vor}
=
O_p\!\left(
\sqrt{\frac{(\ell_n+m_n h_n)\log p_n}{n}}
+
\frac{(\ell_n+m_n h_n)\log p_n}{n}
\right).
\label{eq:proof_t1_stoch2}
\end{equation}

Substitute \eqref{eq:proof_t1_bias} and \eqref{eq:proof_t1_stoch2} into \eqref{eq:proof_t1_decomp1}:
\[
\maxnorm{\Vor-\Vtrue}
=
O_p\!\left(
\ell_n^{-q_0}
+\frac{(m_n+1)h_n}{n}
+\sqrt{\frac{(\ell_n+m_n h_n)\log p_n}{n}}
+\frac{(\ell_n+m_n h_n)\log p_n}{n}
\right).
\]
Then by \eqref{eq:proof_t1_decomp0} and Assumption~\ref{ass:mean_remainder},
\[
\maxnorm{\Vdb-\Vtrue}
=
O_p\!\left(
\maxnorm{\mathbf B_{\mu,n}}
+r_{\mu,n}
+\ell_n^{-q_0}
+\frac{(m_n+1)h_n}{n}
+\sqrt{\frac{(\ell_n+m_n h_n)\log p_n}{n}}
+\frac{(\ell_n+m_n h_n)\log p_n}{n}
\right).
\]
Then Theorem~\ref{thm:maxnorm_rate} follows. \qed

\subsection{Proof of Corollary \ref{cor:balanced_rate}}

Under the additional conditions of the corollary, the terms
\[
\maxnorm{\mathbf B_{\mu,n}}+r_{\mu,n}
\quad\text{and}\quad
\frac{(m_n+1)h_n}{n}
\]
are of smaller order than the leading terms, and \(m_n h_n\lesssim \ell_n\) implies
\[
\ell_n+m_n h_n \asymp \ell_n.
\]
Hence Theorem~\ref{thm:maxnorm_rate} yields
\[
\maxnorm{\Vdb-\Vtrue}
=
O_p\!\left(
\ell_n^{-q_0}
+\sqrt{\frac{\ell_n \log p_n}{n}}
+\frac{\ell_n \log p_n}{n}
\right).
\]
If \(\ell_n\log p_n/n\to 0\), then the last term is dominated by the square-root term, so the leading rate is
\[
\ell_n^{-q_0}+\sqrt{\frac{\ell_n\log p_n}{n}}.
\]
Balancing these two terms gives
\[
\ell_n^{-q_0}\asymp \sqrt{\frac{\ell_n\log p_n}{n}}
\quad\Longrightarrow\quad
\ell_n^{2q_0+1}\asymp \frac{n}{\log p_n}
\quad\Longrightarrow\quad
\ell_n \asymp \left(\frac{n}{\log p_n}\right)^{1/(2q_0+1)}.
\]
Substituting this choice back yields
\[
\ell_n^{-q_0}\asymp \left(\frac{\log p_n}{n}\right)^{q_0/(2q_0+1)},
\qquad
\sqrt{\frac{\ell_n\log p_n}{n}}
\asymp
\left(\frac{\log p_n}{n}\right)^{q_0/(2q_0+1)}.
\]
Therefore \eqref{eq:balanced_rate} follows. \qed

\subsection{Proof of Theorem \ref{thm:thresholding}}
By Theorem~\ref{thm:maxnorm_rate}, there exists a constant \(C_0>0\) such that the event
\[
\mathcal E_n:=\left\{\maxnorm{\Vdb-\Vtrue}\le C_0 a_n\right\}
\]
satisfies \(\Prob(\mathcal E_n)\to 1\). Choose \(\tau_n=C_\tau a_n\) with \(C_\tau\ge 2C_0\). Then on \(\mathcal E_n\),
\begin{equation}
\maxnorm{\Vdb-\Vtrue}\le \tau_n/2.
\label{eq:proof_t2_event}
\end{equation}

Below we work on \(\mathcal E_n\); all bounds then hold with probability tending to one.

Fix \(r\neq s\). Let \(x=\Vdb[r,s]\), \(v=\Vtrue[r,s]\), and \(\widehat v=\Vhard[r,s]=x\1(|x|\ge\tau_n)\). Since \(|x-v|\le \tau_n/2\) on \(\mathcal E_n\), we consider two cases.

\paragraph{Case 1: \(|x|\ge \tau_n\).}
Then \(\widehat v=x\), and
\[
|\widehat v-v|=|x-v|\le \tau_n/2.
\]

\paragraph{Case 2: \(|x|< \tau_n\).}
Then \(\widehat v=0\), and
\[
|\widehat v-v|=|v|
\le |v-x|+|x|
< \tau_n/2+\tau_n = \frac{3}{2}\tau_n.
\]

Thus, for all \(r\neq s\),
\[
|\Vhard[r,s]-\Vtrue[r,s]| \le \frac{3}{2}\tau_n.
\]
For diagonal entries, \(\Vhard[r,r]=\Vdb[r,r]\), so by \eqref{eq:proof_t2_event},
\[
|\Vhard[r,r]-\Vtrue[r,r]|\le \tau_n/2.
\]
Therefore,
\begin{equation}
\maxnorm{\Vhard-\Vtrue}\le \frac{3}{2}\tau_n = O(a_n),
\label{eq:proof_t2_max}
\end{equation}
which proves \eqref{eq:thr_max}.

Fix a row \(r\). We show
\begin{equation}
\sum_{s=1}^{p_n} |\Vhard[r,s]-\Vtrue[r,s]|
\le
C c_{n,p} \tau_n^{1-\alpha}
\label{eq:proof_t2_rowgoal}
\end{equation}
on \(\mathcal E_n\), uniformly in \(r\).

Let \(e_{rs}:=\Vhard[r,s]-\Vtrue[r,s]\). For \(r\neq s\), from the hard-thresholding definition and \(|x-v|\le \tau_n/2\), one has the standard decomposition
\[
|e_{rs}|
\le
|x-v|\1(|x|\ge\tau_n)+|v|\1(|x|<\tau_n).
\]
Hence
\begin{equation}
\sum_{s\neq r}|e_{rs}|
\le
\underbrace{\sum_{s\neq r}|x-v|\1(|x|\ge\tau_n)}_{T_{1r}}
+
\underbrace{\sum_{s\neq r}|v|\1(|x|<\tau_n)}_{T_{2r}}.
\label{eq:proof_t2_rowdecomp}
\end{equation}

On \(\mathcal E_n\), if \(|x|\ge\tau_n\), then \(|v|\ge |x|-|x-v|\ge \tau_n/2\). Therefore
\[
T_{1r}
\le
\frac{\tau_n}{2}\sum_{s\neq r}\1\!\left(|v_{rs}|\ge \frac{\tau_n}{2}\right).
\]
Using the weak row-sparsity condition~\ref{ass:sparsity},
\[
\sum_{s\neq r}\1\!\left(|v_{rs}|\ge \frac{\tau_n}{2}\right)
\le
\left(\frac{2}{\tau_n}\right)^\alpha \sum_{s\neq r}|v_{rs}|^\alpha
\le
C c_{n,p}\tau_n^{-\alpha}.
\]
Hence
\begin{equation}
T_{1r}\le C c_{n,p}\tau_n^{1-\alpha}.
\label{eq:proof_t2_T1}
\end{equation}

On \(\mathcal E_n\), if \(|x|<\tau_n\), then \(|v|<|x|+|x-v|<\tfrac32\tau_n\). Therefore
\[
|v|\1(|x|<\tau_n)
\le
|v|^\alpha \left(\frac{3\tau_n}{2}\right)^{1-\alpha}\1(|x|<\tau_n).
\]
Summing over \(s\neq r\) gives
\begin{equation}
T_{2r}
\le
\left(\frac{3\tau_n}{2}\right)^{1-\alpha}
\sum_{s\neq r}|v_{rs}|^\alpha
\le
C c_{n,p}\tau_n^{1-\alpha}.
\label{eq:proof_t2_T2}
\end{equation}

Combining \eqref{eq:proof_t2_rowdecomp}, \eqref{eq:proof_t2_T1}, and \eqref{eq:proof_t2_T2},
\[
\sum_{s\neq r}|e_{rs}|
\le C c_{n,p}\tau_n^{1-\alpha}.
\]
The diagonal error contributes at most \(\tau_n/2\), which is absorbed into the same order since \(c_{n,p}\ge 1\). Thus \eqref{eq:proof_t2_rowgoal} holds.

Taking the maximum over \(r\), we obtain
\begin{equation}
\onenorm{\Vhard-\Vtrue}
=
\max_{1\le r\le p_n}\sum_{s=1}^{p_n}|e_{rs}|
\le
C c_{n,p}\tau_n^{1-\alpha}.
\label{eq:proof_t2_l1}
\end{equation}

Since \(\Vhard\) and \(\Vtrue\) are symmetric (if desired, one can symmetrize \(\Vdb\) before thresholding by replacing it with \((\Vdb+\widehat{\Vtrue}^{DB\T})/2\), which does not alter the max-norm rate order), we have
\[
\twonorm{\Vhard-\Vtrue}
\le
\onenorm{\Vhard-\Vtrue}.
\]
Using \eqref{eq:proof_t2_l1} and \(\tau_n\asymp a_n\),
\[
\twonorm{\Vhard-\Vtrue}
=
O_p\!\bigl(c_{n,p}a_n^{1-\alpha}\bigr),
\]
which proves \eqref{eq:thr_op}.

For any matrix \(\mathbf A\),
\[
\frobnorm{\mathbf A}^2
=
\sum_{r,s} A_{rs}^2
\le
\max_{r,s}|A_{rs}|
\sum_{r,s}|A_{rs}|
\le
p_n \maxnorm{\mathbf A}\,\onenorm{\mathbf A}.
\]
Applying this to \(\mathbf A=\Vhard-\Vtrue\), together with \eqref{eq:proof_t2_max} and \eqref{eq:proof_t2_l1}, yields
\[
\frac{1}{p_n}\frobnorm{\Vhard-\Vtrue}^2
\le
\maxnorm{\Vhard-\Vtrue}\,\onenorm{\Vhard-\Vtrue}
\le
C a_n \cdot C c_{n,p} a_n^{1-\alpha}
=
C c_{n,p} a_n^{2-\alpha}.
\]
Thus \eqref{eq:thr_frob} follows.
 \qed

\subsection{Proof of Theorem \ref{thm:soft_thresholding}}
\label{app:proof_soft_thresholding}

We parallel the proof of Theorem~\ref{thm:thresholding}, but now for the soft-thresholding map
\[
S_\tau(x):=\sign(x)(|x|-\tau)_+.
\]

By Theorem~\ref{thm:maxnorm_rate}, there exists \(C_0>0\) such that
\[
\mathcal E_n:=\left\{\maxnorm{\Vdb-\Vtrue}\le C_0 a_n\right\}
\]
satisfies \(\Prob(\mathcal E_n)\to 1\). Choose \(\tau_n=C_\tau a_n\) with \(C_\tau\ge 2C_0\). Then on \(\mathcal E_n\),
\begin{equation*}
\maxnorm{\Vdb-\Vtrue}\le \tau_n/2.
\end{equation*}

We work on \(\mathcal E_n\) throughout.

Fix \(r\neq s\). Let \(x=\Vdb[r,s]\), \(v=\Vtrue[r,s]\), and \(\widehat v=S_{\tau_n}(x)\). Consider two cases.

\emph{Case 1: \(|x|\ge\tau_n\).} Then \(\widehat v = x-\tau_n\sign(x)\), so
\[
|\widehat v-v|
\le |x-v|+\tau_n
\le \frac{3}{2}\tau_n.
\]

\emph{Case 2: \(|x|<\tau_n\).} Then \(\widehat v=0\), and
\[
|\widehat v-v|=|v|\le |v-x|+|x|<\frac{\tau_n}{2}+\tau_n = \frac{3}{2}\tau_n.
\]

Thus, for all off-diagonal entries,
\[
|\Vsoft[r,s]-\Vtrue[r,s]|\le \frac{3}{2}\tau_n.
\]
Diagonal entries are not thresholded, hence
\[
|\Vsoft[r,r]-\Vtrue[r,r]|=|\Vdb[r,r]-\Vtrue[r,r]|\le \tau_n/2.
\]
Therefore
\[
\maxnorm{\Vsoft-\Vtrue}\le \frac{3}{2}\tau_n = O(a_n),
\]
which proves \eqref{eq:soft_max}.

Fix a row \(r\). Let \(e_{rs}:=\Vsoft[r,s]-\Vtrue[r,s]\). For \(r\neq s\), note that
\[
|S_{\tau_n}(x)-v|
\le
|x-v| + \tau_n \1(|x|\ge\tau_n) + |v|\1(|x|<\tau_n),
\]
because when \(|x|\ge\tau_n\), the soft-thresholding shrinkage introduces an extra \(\tau_n\), and when \(|x|<\tau_n\), the estimator is zero.

Hence
\begin{align*}
\sum_{s\neq r}|e_{rs}|
&\le
\underbrace{\sum_{s\neq r}|x_{rs}-v_{rs}|}_{A_{1r}}
+
\underbrace{\tau_n\sum_{s\neq r}\1(|x_{rs}|\ge\tau_n)}_{A_{2r}}
+
\underbrace{\sum_{s\neq r}|v_{rs}|\1(|x_{rs}|<\tau_n)}_{A_{3r}}.
\end{align*}

The first term satisfies \(A_{1r}\le p_n(\tau_n/2)\), which is too crude if taken directly. Instead, refine it by restricting to indices with \(|x_{rs}|\ge\tau_n\), since the universal \(|x-v|\) term only matters in the first case:
\[
|S_{\tau_n}(x)-v|
\le
\Big(|x-v|+\tau_n\Big)\1(|x|\ge\tau_n)
+
|v|\1(|x|<\tau_n).
\]
Thus
\begin{equation*}
\sum_{s\neq r}|e_{rs}|
\le
\underbrace{\sum_{s\neq r}|x_{rs}-v_{rs}|\1(|x_{rs}|\ge\tau_n)}_{T_{1r}}
+
\underbrace{\tau_n\sum_{s\neq r}\1(|x_{rs}|\ge\tau_n)}_{T_{2r}}
+
\underbrace{\sum_{s\neq r}|v_{rs}|\1(|x_{rs}|<\tau_n)}_{T_{3r}}.
\end{equation*}

On \(\mathcal E_n\), if \(|x_{rs}|\ge\tau_n\), then \(|v_{rs}|\ge \tau_n/2\). Hence
\[
T_{1r}\le \frac{\tau_n}{2}\sum_{s\neq r}\1(|v_{rs}|\ge\tau_n/2),
\qquad
T_{2r}\le \tau_n\sum_{s\neq r}\1(|v_{rs}|\ge\tau_n/2).
\]
By weak row-sparsity,
\[
\sum_{s\neq r}\1(|v_{rs}|\ge\tau_n/2)
\le
\left(\frac{2}{\tau_n}\right)^\alpha\sum_{s\neq r}|v_{rs}|^\alpha
\le
C c_{n,p}\tau_n^{-\alpha}.
\]
Therefore,
\[
T_{1r}+T_{2r}\le C c_{n,p}\tau_n^{1-\alpha}.
\]

For \(T_{3r}\), on \(\mathcal E_n\), \(|x_{rs}|<\tau_n\) implies \(|v_{rs}|<3\tau_n/2\). Hence
\[
|v_{rs}|\1(|x_{rs}|<\tau_n)
\le
|v_{rs}|^\alpha \left(\frac{3\tau_n}{2}\right)^{1-\alpha},
\]
and summing gives
\[
T_{3r}\le C c_{n,p}\tau_n^{1-\alpha}.
\]
Combining,
\[
\sum_{s\neq r}|e_{rs}|
\le C c_{n,p}\tau_n^{1-\alpha}.
\]
The diagonal contribution is \(O(\tau_n)\), absorbed by the same order since \(c_{n,p}\ge 1\). Therefore
\[
\onenorm{\Vsoft-\Vtrue}\le C c_{n,p}\tau_n^{1-\alpha}.
\]

The operator norm bound follows exactly as in Theorem~\ref{thm:thresholding}:
\[
\twonorm{\Vsoft-\Vtrue}
\le
\onenorm{\Vsoft-\Vtrue}
=
O_p(c_{n,p}a_n^{1-\alpha}).
\]
For the Frobenius norm,
\[
\frac1{p_n}\frobnorm{\Vsoft-\Vtrue}^2
\le
\maxnorm{\Vsoft-\Vtrue}\,\onenorm{\Vsoft-\Vtrue}
=
O_p(a_n)\cdot O_p(c_{n,p}a_n^{1-\alpha})
=
O_p(c_{n,p}a_n^{2-\alpha}).
\]
This proves \eqref{eq:soft_op} and \eqref{eq:soft_frob}. \qed

\subsection{Proof of Theorem \ref{thm:tapering}}
\label{app:proof_tapering}

Let
\[
\mathbf E_n := \Vdb-\Vtrue.
\]
Then the tapering estimator satisfies
\[
\mathbf V^{\mathrm{tap}}-\Vtrue
=
\mathbf W^{(k_n)}\circ \mathbf E_n
+
\bigl(\mathbf W^{(k_n)}-\mathbf 1\mathbf 1^\T\bigr)\circ \Vtrue.
\]
We bound the two terms separately.

By Theorem~\ref{thm:maxnorm_rate}, \(\maxnorm{\mathbf E_n}=O_p(a_n)\). On the corresponding high-probability event, for each row \(r\),
\[
\sum_{s=1}^{p_n}\left|w_{rs}^{(k_n)} E_{n,rs}\right|
\le
\max_{r,s}|E_{n,rs}|\sum_{s=1}^{p_n} w_{rs}^{(k_n)}
\le
C_w k_n \maxnorm{\mathbf E_n}
\le
C k_n a_n.
\]
Hence
\begin{equation}
\onenorm{\mathbf W^{(k_n)}\circ \mathbf E_n}
\le C k_n a_n
\quad\text{and thus}\quad
\twonorm{\mathbf W^{(k_n)}\circ \mathbf E_n}\le C k_n a_n,
\label{eq:taper_estimation_op}
\end{equation}
using symmetry.

For the Frobenius norm, each row has at most \(C_w k_n\) nonzero taper weights and \(0\le w_{rs}^{(k_n)}\le 1\), so
\begin{equation}
\frac1{p_n}\frobnorm{\mathbf W^{(k_n)}\circ \mathbf E_n}^2
\le
\frac1{p_n}\cdot p_n \cdot C_w k_n \cdot \maxnorm{\mathbf E_n}^2
=
O_p(k_n a_n^2).
\label{eq:taper_estimation_frob}
\end{equation}

By Definition~\ref{def:tapering}, \(w_{rs}^{(k_n)}=1\) whenever \(|r-s|\le k_n/2\), so the approximation error is supported on \(|r-s|>k_n/2\). Since \(0\le w_{rs}^{(k_n)}\le 1\),
\[
\left|
\left[(\mathbf W^{(k_n)}-\mathbf 1\mathbf 1^\T)\circ \Vtrue\right]_{rs}
\right|
\le
|V_{rs}|\,\1(|r-s|>k_n/2).
\]
Therefore, by Assumption~\ref{ass:bandable},
\begin{align*}
\onenorm{(\mathbf W^{(k_n)}-\mathbf 1\mathbf 1^\T)\circ \Vtrue}
&\le
\max_{1\le r\le p_n}\sum_{|r-s|>k_n/2}|V_{rs}|
\notag\\
&\le
C_{\mathrm{bd}}\left(\frac{k_n}{2}\right)^{-\nu}
\le
C k_n^{-\nu}.
\end{align*}
Hence
\begin{equation}
\twonorm{(\mathbf W^{(k_n)}-\mathbf 1\mathbf 1^\T)\circ \Vtrue}
\le
C k_n^{-\nu}.
\label{eq:taper_approx_op}
\end{equation}

Similarly, by \eqref{eq:bandable_tail_l2},
\begin{align}
\frac1{p_n}\frobnorm{(\mathbf W^{(k_n)}-\mathbf 1\mathbf 1^\T)\circ \Vtrue}^2
&\le
\frac1{p_n}\sum_{r=1}^{p_n}\sum_{|r-s|>k_n/2}V_{rs}^2
\notag\\
&\le
C_{\mathrm{bd},2}\left(\frac{k_n}{2}\right)^{-2\nu}
\le
C k_n^{-2\nu}.
\label{eq:taper_approx_frob}
\end{align}

Using the triangle inequality with \eqref{eq:taper_estimation_op} and \eqref{eq:taper_approx_op},
\[
\twonorm{\mathbf V^{\mathrm{tap}}-\Vtrue}
\le
\twonorm{\mathbf W^{(k_n)}\circ \mathbf E_n}
+
\twonorm{(\mathbf W^{(k_n)}-\mathbf 1\mathbf 1^\T)\circ \Vtrue}
=
O_p(k_n a_n + k_n^{-\nu}),
\]
which proves \eqref{eq:taper_op}.

For the Frobenius norm, by \((a+b)^2\le 2a^2+2b^2\), together with \eqref{eq:taper_estimation_frob} and \eqref{eq:taper_approx_frob},
\[
\frac1{p_n}\frobnorm{\mathbf V^{\mathrm{tap}}-\Vtrue}^2
=
O_p(k_n a_n^2 + k_n^{-2\nu}),
\]
proving \eqref{eq:taper_frob}.

Finally, the balancing choice \(k_n\asymp a_n^{-1/(\nu+1)}\) equalizes \(k_n a_n\) and \(k_n^{-\nu}\), yielding
\[
\twonorm{\mathbf V^{\mathrm{tap}}-\Vtrue}
=
O_p\!\left(a_n^{\nu/(\nu+1)}\right).
\]
This completes the proof. \qed

\section{Proofs for the explicit mean-term bounds and additional refinements}

\subsection{Proof of Proposition \ref{prop:Bmu_generic}}
\label{app:proof_Bmu_generic}

Recall
\[
\mathbf B_{\mu,n}
=
\sum_{|k|<\ell_n} K\!\left(\frac{k}{\ell_n}\right)
\left[
\frac{1}{n}\sum_{t=m_n h_n+|k|+1}^{n} \mathbf M_t \mathbf M_{t-|k|}^\T
\right].
\]
Fix any entry \((r,s)\). By the triangle inequality and \(|K(\cdot)|\le \norm{K}_{\infty}\),
\begin{align*}
|B_{\mu,n}[r,s]|
&\le
\sum_{|k|<\ell_n}\norm{K}_{\infty}
\cdot
\frac{1}{n}\sum_{t=m_n h_n+|k|+1}^{n}
|M_t^{(r)} M_{t-|k|}^{(s)}| \\
&\le
\sum_{|k|<\ell_n}\norm{K}_{\infty}
\cdot
\frac{1}{2n}\sum_{t=m_n h_n+|k|+1}^{n}
\left\{ |M_t^{(r)}|^2 + |M_{t-|k|}^{(s)}|^2 \right\}\\
&\le
\sum_{|k|<\ell_n}\norm{K}_{\infty}
\cdot
\frac{1}{n}\sum_{u=m_n h_n+1}^{n}\norm{\mathbf M_u}_{\max}^2 \\
&=
(2\ell_n-1)\norm{K}_{\infty}\,\bar M_{2,n}.
\end{align*}
Taking the maximum over \((r,s)\) proves proposition~\ref{prop:Bmu_generic}. \qed

\subsection{Proof of Proposition \ref{prop:Rmu_generic}}
\label{app:proof_Rmu_generic}

By Assumption~\ref{ass:cross_bernstein}, conditional on \(\{\bm\mu_t\}_{t=1}^n\),
\[
\Prob\!\left(
\maxnorm{\mathbf R_{\mu,n}} > t
\ \middle|\ \{\bm\mu_t\}
\right)
\le
4p_n^2 \exp\!\left[
-c_3 n \min\!\left\{
\frac{t^2}{L_n^2 \bar M_{2,n}},
\frac{t}{L_n \bar M_{\infty,n}}
\right\}
\right].
\]
We now choose \(t\) to offset the factor \(p_n^2\).

Let
\[
t_1 := A_1 \sqrt{\bar M_{2,n}}\,L_n\sqrt{\frac{\log p_n}{n}},
\qquad
t_2 := A_2 \bar M_{\infty,n}\frac{L_n\log p_n}{n},
\]
where \(A_1,A_2>0\) are large constants. Then
\[
\frac{n t_1^2}{L_n^2 \bar M_{2,n}} = A_1^2 \log p_n,
\qquad
\frac{n t_2}{L_n \bar M_{\infty,n}} = A_2 \log p_n.
\]
By the union-bound-adjusted Bernstein form, choosing \(A_1,A_2\) sufficiently large gives
\[
\Prob\!\left(
\maxnorm{\mathbf R_{\mu,n}} > C(t_1+t_2)
\ \middle|\ \{\bm\mu_t\}
\right)
\to 0
\]
for a universal constant \(C\). Integrating out the conditioning yields
\[
\maxnorm{\mathbf R_{\mu,n}}
=
O_p(t_1+t_2),
\]
which is exactly \eqref{eq:Rmu_generic_bound}. The choice \eqref{eq:rmu_generic_choice} follows immediately. \qed

\subsection{Proof of Proposition \ref{prop:holder_mean_bounds}}
\label{app:proof_holder_mean_bounds}

We prove the claims in order.
Using \(\sum_{j=0}^{m_n} d_{n,j}=0\),
\begin{equation}
\mathbf M_t
=
\sum_{j=0}^{m_n} d_{n,j}\bm\mu_{t-jh_n}
=
\sum_{j=1}^{m_n} d_{n,j}\left(\bm\mu_{t-jh_n}-\bm\mu_t\right).
\label{eq:Mt_zero_sum_repr}
\end{equation}

Let
\[
\mathcal B_n
:=
\left\{
t\in\{m_n h_n+1,\ldots,n\}:
[t-m_n h_n,t]\ \text{crosses at least one partition boundary } \eta_j n
\right\}.
\]
Each boundary can contaminate at most \(m_n h_n+1\) indices, so
\begin{equation}
|\mathcal B_n|\le C J_n m_n h_n.
\label{eq:Bn_cardinality}
\end{equation}
For \(t\notin\mathcal B_n\), all indices \(t-jh_n\), \(j=0,\ldots,m_n\), lie in the same H\"older segment.

For \(t\notin\mathcal B_n\), by \eqref{eq:Mt_zero_sum_repr} and the H\"older condition,
\begin{align*}
\norm{\mathbf M_t}_{\max}
&\le
\sum_{j=1}^{m_n}|d_{n,j}|\,
\norm{\bm\mu_{t-jh_n}-\bm\mu_t}_{\max}
\notag\\
&\le
L_\mu\sum_{j=1}^{m_n}|d_{n,j}|\left(\frac{j h_n}{n}\right)^\beta
=
L_\mu C_{d,\beta}\left(\frac{h_n}{n}\right)^\beta.
\end{align*}

For any \(t\), using \(\sup_u\norm{\bm\mu_u}_{\max}\le M_\mu\),
\[
\norm{\mathbf M_t}_{\max}
\le
\sum_{j=0}^{m_n}|d_{n,j}|\norm{\bm\mu_{t-jh_n}}_{\max}
\le
C_{d,1}M_\mu.
\]

By splitting the sum over \(t\) into \(\mathcal B_n\) and \(\mathcal B_n^c\),
\begin{align*}
\bar M_{2,n}
&=
\frac1n\sum_{t=m_n h_n+1}^{n}\norm{\mathbf M_t}_{\max}^2 \\
&\le
\frac1n\sum_{t\notin\mathcal B_n}
\left\{
L_\mu C_{d,\beta}\left(\frac{h_n}{n}\right)^\beta
\right\}^2
+
\frac1n\sum_{t\in\mathcal B_n}
\left(C_{d,1}M_\mu\right)^2 \\
&\le
L_\mu^2 C_{d,\beta}^2\left(\frac{h_n}{n}\right)^{2\beta}
+
C_{d,1}^2M_\mu^2\,\frac{|\mathcal B_n|}{n}.
\end{align*}
Using \eqref{eq:Bn_cardinality} gives \eqref{eq:Mbar2_holder}.

Taking maxima in the good/bad bounds,
\[
\bar M_{\infty,n}
\le
\max\!\left\{
L_\mu C_{d,\beta}\left(\frac{h_n}{n}\right)^\beta,\,
C_{d,1}M_\mu
\right\}
\le
C\left\{
C_{d,\beta}\left(\frac{h_n}{n}\right)^\beta + C_{d,1}M_\mu
\right\},
\]
which proves \eqref{eq:Mbarinf_holder}.

Equation \eqref{eq:Bmu_holder} follows immediately from Proposition~\ref{prop:Bmu_generic} and \eqref{eq:Mbar2_holder}. Equation \eqref{eq:Rmu_holder} follows from Proposition~\ref{prop:Rmu_generic}, \eqref{eq:Mbar2_holder}, and \eqref{eq:Mbarinf_holder}. \qed

\subsection{Proof of Proposition \ref{prop:bv_mean_bounds}}
\label{app:proof_bv_mean_bounds}

Let \(\Delta\bm\mu_u:=\bm\mu_u-\bm\mu_{u-1}\) for \(u=2,\ldots,n\). By Assumption~\ref{ass:mean_bv},
\[
\sum_{u=2}^{n}\norm{\Delta\bm\mu_u}_{\max}
\le
\sum_{u=2}^{n}\max_{1\le r\le p_n}|\mu_u^{(r)}-\mu_{u-1}^{(r)}|
\le
\max_{1\le r\le p_n}\sum_{u=2}^{n}|\mu_u^{(r)}-\mu_{u-1}^{(r)}|
\le V_{\mu,n}.
\]

Using \(\sum_{j=0}^{m_n}d_{n,j}=0\) as in \eqref{eq:Mt_zero_sum_repr},
\[
\norm{\mathbf M_t}_{\max}
\le
\sum_{j=1}^{m_n}|d_{n,j}|\,\norm{\bm\mu_t-\bm\mu_{t-jh_n}}_{\max}
\le
C_{d,1}\max_{1\le j\le m_n}\norm{\bm\mu_t-\bm\mu_{t-jh_n}}_{\max}.
\]
For each \(j\le m_n\),
\[
\norm{\bm\mu_t-\bm\mu_{t-jh_n}}_{\max}
\le
\sum_{u=t-jh_n+1}^{t}\norm{\Delta\bm\mu_u}_{\max}
\le
\sum_{u=t-m_n h_n+1}^{t}\norm{\Delta\bm\mu_u}_{\max}.
\]
Hence
\begin{equation*}
\norm{\mathbf M_t}_{\max}
\le
C_{d,1}\sum_{u=t-m_n h_n+1}^{t}\norm{\Delta\bm\mu_u}_{\max}.
\end{equation*}

Squaring and using Cauchy--Schwarz,
\[
\left(\sum_{u=t-m_n h_n+1}^{t}a_u\right)^2
\le
(m_n h_n)\sum_{u=t-m_n h_n+1}^{t} a_u^2,
\]
with \(a_u=\norm{\Delta\bm\mu_u}_{\max}\). Therefore
\[
\norm{\mathbf M_t}_{\max}^2
\le
C_{d,1}^2 (m_n h_n)\sum_{u=t-m_n h_n+1}^{t}\norm{\Delta\bm\mu_u}_{\max}^2.
\]
Summing over \(t\) and dividing by \(n\), each \(\norm{\Delta\bm\mu_u}_{\max}^2\) is counted at most \(m_n h_n\) times, so
\begin{align*}
\bar M_{2,n}
&\le
\frac{C_{d,1}^2 (m_n h_n)^2}{n}\sum_{u=2}^{n}\norm{\Delta\bm\mu_u}_{\max}^2 \\
&\le
\frac{C_{d,1}^2 (m_n h_n)^2}{n}
\left(\sum_{u=2}^{n}\norm{\Delta\bm\mu_u}_{\max}\right)^2
\le
C\,C_{d,1}^2 (m_n h_n)^2\frac{V_{\mu,n}^2}{n},
\end{align*}
which proves \eqref{eq:Mbar2_bv}.

Next, \(\bar M_{\infty,n}\le \sum_{j=0}^{m_n}|d_{n,j}|\sup_t\norm{\bm\mu_t}_{\max}\le C_{d,1}M_\mu\), and enlarging constants yields \eqref{eq:Mbarinf_bv}. The bounds \eqref{eq:Bmu_bv} and \eqref{eq:Rmu_bv} then follow from Propositions~\ref{prop:Bmu_generic} and \ref{prop:Rmu_generic}. \qed

\section{A unified sufficient condition for Assumptions \ref{ass:oracle_bias}, \ref{ass:bernstein} and \ref{ass:cross_bernstein}}

\begin{assumption}[A unified sufficient condition]
\label{ass:suff_unified}
Assume that the noise process admits a stable linear representation
\[
\mathbf Z_t=\sum_{a=0}^{\infty}\mathbf A_a\,\boldsymbol\varepsilon_{t-a},
\qquad t\in\mathbb Z,
\]
where the innovations $\{\boldsymbol\varepsilon_t\}_{t\in\mathbb Z}$ are i.i.d.\ in
$\mathbb R^{p_n}$ and satisfy the following conditions.

\begin{enumerate}
\item[(U1)] (\emph{Independent sub-Gaussian coordinates})
For each $t$, the coordinates of $\boldsymbol\varepsilon_t
=(\varepsilon_{t,1},\ldots,\varepsilon_{t,p_n})^\top$ are independent,
$\E(\boldsymbol\varepsilon_t)=\mathbf 0$, and
\[
\max_{1\le j\le p_n}\|\varepsilon_{0,j}\|_{\psi_2}\le K_\varepsilon<\infty.
\]

\item[(U2)] (\emph{Stable and polynomially decaying filter})
There exist constants $K_A<\infty$, $\widetilde K_A<\infty$, and $\eta>0$ such that
\[
\sum_{a=0}^{\infty}\|\mathbf A_a\|_{\mathrm{op}}\le K_A,
\]
and
\[
\|\mathbf A_a\|_{\mathrm{op}}
\le \widetilde K_A (1+a)^{-(q_0+2+\eta)},
\qquad a\ge0.
\]

\item[(U3)] (\emph{Kernel and difference weights})
The kernel is bounded on $[-1,1]$:
\[
\sup_{|x|\le1}|K(x)|\le K_0,
\]
and the difference sequence satisfies
\[
\sum_{j=0}^{m_n}|d_{n,j}|\le K_d
\]
in addition to \eqref{eq:diff_normalized}.

\item[(U4)] (\emph{Bandwidth/spacing rates})
\[
\ell_n=o(h_n),\qquad
\frac{h_n^{\,q_0+1+\eta}}{\ell_n^{\,q_0+1}}\to\infty,
\qquad
\ell_n\le h_n.
\]
\end{enumerate}
\end{assumption}

\begin{lemma}
\label{lem:basic_unified}
Under Assumption~\ref{ass:suff_unified}, for every fixed $q\ge4$, the following hold.

\begin{enumerate}
\item[(i)] (\emph{Uniform moment bound})
There exists a constant $C_q>0$, depending only on
$q,K_\varepsilon,K_A$, such that
\[
\max_{1\le r\le p_n}\|Z_{0,r}\|_q\le C_q.
\]

\item[(ii)] (\emph{Physical dependence bound})
Let $\mathbf Z_k^{(0)}$ be the coupled version obtained by replacing
$\boldsymbol\varepsilon_0$ with an i.i.d.\ copy $\boldsymbol\varepsilon_0'$.
Then
\[
\delta_q(k):=\max_{1\le r\le p_n}\|Z_{k,r}-Z_{k,r}^{(0)}\|_q
\le C_q (1+k)^{-(q_0+2+\eta)},\qquad k\ge1.
\]

\item[(iii)] (\emph{Covariance decay})
Let $\Gamma_k=\E(\mathbf Z_0\mathbf Z_k^\top)$. Then there exists $C_\Gamma>0$ such that
\[
\|\Gamma_k\|_{\max}\le C_\Gamma (1+|k|)^{-(q_0+2+\eta)},
\qquad k\in\mathbb Z.
\]
Consequently,
\[
\sum_{k\in\mathbb Z}\|\Gamma_k\|_{\max}<\infty,
\qquad
\sum_{k\in\mathbb Z}(1+|k|)^{q_0}\|\Gamma_k\|_{\max}<\infty.
\]
\end{enumerate}
\end{lemma}

\begin{proof}
We prove the three statements in turn.

\medskip\noindent
\textbf{Proof of (i).}
Fix $r$. Since
\[
Z_{0,r}=\sum_{a=0}^{\infty} e_r^\top \mathbf A_a \boldsymbol\varepsilon_{-a},
\]
by Minkowski's inequality,
\[
\|Z_{0,r}\|_q
\le
\sum_{a=0}^{\infty}
\|e_r^\top \mathbf A_a \boldsymbol\varepsilon_{-a}\|_q.
\]
Because the coordinates of $\boldsymbol\varepsilon_{-a}$ are independent and uniformly
sub-Gaussian, every linear form is sub-Gaussian, hence for each fixed $q\ge2$,
\[
\|e_r^\top \mathbf A_a \boldsymbol\varepsilon_{-a}\|_q
\le C\sqrt q\,K_\varepsilon\,\| \mathbf A_a^\top e_r\|_2
\le C\sqrt q\,K_\varepsilon\,\|\mathbf A_a\|_{\mathrm{op}}.
\]
Therefore,
\[
\|Z_{0,r}\|_q
\le
C\sqrt q\,K_\varepsilon \sum_{a=0}^{\infty}\|\mathbf A_a\|_{\mathrm{op}}
\le C\sqrt q\,K_\varepsilon K_A.
\]
Taking the maximum over $r$ proves (i).

\medskip\noindent
\textbf{Proof of (ii).}
For $k\ge1$,
\[
\mathbf Z_k-\mathbf Z_k^{(0)}
=
\mathbf A_k(\boldsymbol\varepsilon_0-\boldsymbol\varepsilon_0').
\]
Hence, for each coordinate $r$,
\[
Z_{k,r}-Z_{k,r}^{(0)}
=
e_r^\top \mathbf A_k(\boldsymbol\varepsilon_0-\boldsymbol\varepsilon_0').
\]
Applying again the linear-form bound for independent sub-Gaussian coordinates yields
\[
\|Z_{k,r}-Z_{k,r}^{(0)}\|_q
\le
C\sqrt q\,K_\varepsilon\,\|\mathbf A_k\|_{\mathrm{op}}.
\]
Using (U2),
\[
\delta_q(k)
\le
C\sqrt q\,K_\varepsilon\,\widetilde K_A (1+k)^{-(q_0+2+\eta)}.
\]
This proves (ii).

\medskip\noindent
\textbf{Proof of (iii).}
Let
\[
\Sigma_\varepsilon:=\E(\boldsymbol\varepsilon_0\boldsymbol\varepsilon_0^\top).
\]
Since the coordinates of $\boldsymbol\varepsilon_0$ are independent and centered,
$\Sigma_\varepsilon$ is diagonal, and
\[
\|\Sigma_\varepsilon\|_{\mathrm{op}}
=
\max_{1\le j\le p_n}\Var(\varepsilon_{0,j})
\le
C K_\varepsilon^2.
\]
For $k\ge0$, by the linear-process representation and independence across times,
\[
\Gamma_k
=
\sum_{a=0}^{\infty}
\mathbf A_a \Sigma_\varepsilon \mathbf A_{a+k}^\top.
\]
Therefore,
\[
\|\Gamma_k\|_{\max}
\le
\|\Gamma_k\|_{\mathrm{op}}
\le
\sum_{a=0}^{\infty}
\|\mathbf A_a\|_{\mathrm{op}}\,
\|\Sigma_\varepsilon\|_{\mathrm{op}}\,
\|\mathbf A_{a+k}\|_{\mathrm{op}}.
\]
Using $\sup_{a\ge0}\|\mathbf A_{a+k}\|_{\mathrm{op}}
\le \widetilde K_A(1+k)^{-(q_0+2+\eta)}$ and
$\sum_{a\ge0}\|\mathbf A_a\|_{\mathrm{op}}\le K_A$, we get
\[
\|\Gamma_k\|_{\max}
\le
C K_\varepsilon^2 K_A \widetilde K_A (1+k)^{-(q_0+2+\eta)}.
\]
For negative $k$, use $\Gamma_{-k}=\Gamma_k^\top$, so the same bound holds with $|k|$.
Hence
\[
\|\Gamma_k\|_{\max}\le C_\Gamma (1+|k|)^{-(q_0+2+\eta)}.
\]
Since $q_0+2+\eta>q_0+1$, both series
\[
\sum_{k\in\mathbb Z}\|\Gamma_k\|_{\max},
\qquad
\sum_{k\in\mathbb Z}(1+|k|)^{q_0}\|\Gamma_k\|_{\max}
\]
converge. This proves (iii).
\end{proof}

\begin{proposition}[Entrywise Bernstein concentration for the oracle estimator]
\label{prop:bernstein_unified}
Under Assumption~\ref{ass:suff_unified}, there exist constants $c_1,c_2>0$
depending only on $K_\varepsilon,K_A,K_0,K_d$ such that for all sufficiently large $n$,
all $1\le r,s\le p_n$, and all $t>0$,
\[
\Prob\!\left(\big|\Vor[r,s]-\E\Vor[r,s]\big|>t\right)
\le
2\exp\!\left[
-c_1 n \min\!\left\{
\frac{t^2}{\ell_n},\,
\frac{t}{\ell_n}
\right\}
\right].
\]
In particular, Assumption~\ref{ass:bernstein} holds, because
\[
\ell_n\le \ell_n+m_n h_n.
\]
\end{proposition}

\begin{proof}
Fix $1\le r,s\le p_n$.
We first prove the result for a truncated linear process and then let the truncation level tend to infinity.

For $M\ge0$, define
\[
\mathbf Z_t^{(M)}:=\sum_{a=0}^{M}\mathbf A_a \boldsymbol\varepsilon_{t-a},
\qquad
\mathbf U_t^{(M)}:=\sum_{j=0}^{m_n} d_{n,j}\mathbf Z_{t-jh_n}^{(M)}.
\]
Let
\[
{\hat{\mathbf{V}}^{or(M)}}[r,s]
\]
be the oracle estimator obtained from $\{\mathbf U_t^{(M)}\}$ in place of $\{\mathbf U_t\}$.

Write
\[
\mathbf z_r^{(M)}:=(Z_{1,r}^{(M)},\ldots,Z_{n,r}^{(M)})^\top,
\qquad
\mathbf z_s^{(M)}:=(Z_{1,s}^{(M)},\ldots,Z_{n,s}^{(M)})^\top.
\]
Also let
\[
\mathbf u_r^{(M)}:=(U_{1,r}^{(M)},\ldots,U_{n,r}^{(M)})^\top,
\qquad
\mathbf u_s^{(M)}:=(U_{1,s}^{(M)},\ldots,U_{n,s}^{(M)})^\top.
\]

Define the $n\times n$ matrix
\[
(\mathbf H_n)_{tu}
:=
K\!\left(\frac{t-u}{\ell_n}\right)\mathbf 1\{|t-u|<\ell_n\},
\qquad 1\le t,u\le n.
\]
Then, up to the usual boundary convention (equivalently, by extending the series by zero
outside the valid range),
\[
{\hat{\mathbf{V}}^{or(M)}}[r,s]
=
\frac1n (\mathbf u_r^{(M)})^\top \mathbf H_n \mathbf u_s^{(M)}.
\]

Next define the $n\times n$ linear filter matrix
\[
(\mathbf F_n)_{tv}
:=
\sum_{j=0}^{m_n} d_{n,j}\,\mathbf 1\{v=t-jh_n\},
\qquad 1\le t,v\le n,
\]
so that
\[
\mathbf u_r^{(M)}=\mathbf F_n \mathbf z_r^{(M)},
\qquad
\mathbf u_s^{(M)}=\mathbf F_n \mathbf z_s^{(M)}.
\]

Now collect the innovations into the finite vector
\[
\boldsymbol\xi_M
:=
(\boldsymbol\varepsilon_{1-M}^\top,\ldots,\boldsymbol\varepsilon_n^\top)^\top
\in\mathbb R^{(n+M)p_n}.
\]
There exist matrices $\mathbf G_{r,M}$ and $\mathbf G_{s,M}$ such that
\[
\mathbf z_r^{(M)}=\mathbf G_{r,M}\boldsymbol\xi_M,
\qquad
\mathbf z_s^{(M)}=\mathbf G_{s,M}\boldsymbol\xi_M.
\]
Hence
\[
\mathbf u_r^{(M)}=\mathbf L_{r,M}\boldsymbol\xi_M,
\qquad
\mathbf u_s^{(M)}=\mathbf L_{s,M}\boldsymbol\xi_M,
\]
where
\[
\mathbf L_{r,M}:=\mathbf F_n\mathbf G_{r,M},
\qquad
\mathbf L_{s,M}:=\mathbf F_n\mathbf G_{s,M}.
\]
Therefore,
\[
{\hat{\mathbf{V}}^{or(M)}}[r,s]
=
\frac1n\,\boldsymbol\xi_M^\top \mathbf Q_M \boldsymbol\xi_M,
\]
where
\[
\mathbf Q_M
:=
\frac12\Big(
\mathbf L_{r,M}^\top \mathbf H_n \mathbf L_{s,M}
+
\mathbf L_{s,M}^\top \mathbf H_n^\top \mathbf L_{r,M}
\Big).
\]
The matrix $\mathbf Q_M$ is symmetric, so Hanson--Wright applies directly.

First, $\mathbf F_n$ satisfies
\[
\|\mathbf F_n\|_{\infty}
\le \sum_{j=0}^{m_n}|d_{n,j}|
\le K_d,
\qquad
\|\mathbf F_n\|_{1}
\le \sum_{j=0}^{m_n}|d_{n,j}|
\le K_d.
\]
Hence
\[
\|\mathbf F_n\|_{\mathrm{op}}
\le
\sqrt{\|\mathbf F_n\|_1\|\mathbf F_n\|_\infty}
\le K_d.
\]

Second, since $\mathbf H_n$ has at most $2\ell_n-1$ nonzero entries in each row and each column,
and each nonzero entry is bounded by $K_0$, we have
\[
\|\mathbf H_n\|_{\infty}\le (2\ell_n-1)K_0,
\qquad
\|\mathbf H_n\|_{1}\le (2\ell_n-1)K_0,
\]
so
\[
\|\mathbf H_n\|_{\mathrm{op}}
\le
\sqrt{\|\mathbf H_n\|_1\|\mathbf H_n\|_\infty}
\le C\ell_n.
\]
Moreover, $\mathbf H_n$ has at most $n(2\ell_n-1)$ nonzero entries, each of magnitude at most $K_0$,
so
\[
\|\mathbf H_n\|_F^2\le C n\ell_n.
\]

Third, we bound $\mathbf G_{r,M}$. Let
\[
x=(x_{1-M}^\top,\ldots,x_n^\top)^\top \in \mathbb R^{(n+M)p_n},
\qquad x_\tau\in\mathbb R^{p_n}.
\]
Then
\[
(\mathbf G_{r,M}x)_t=\sum_{a=0}^{M} e_r^\top \mathbf A_a x_{t-a},
\qquad 1\le t\le n.
\]
Hence, by the triangle inequality in $\ell_2$,
\[
\|\mathbf G_{r,M}x\|_2
=
\left\|
\sum_{a=0}^{M}
\big(e_r^\top\mathbf A_a x_{1-a},\ldots,e_r^\top\mathbf A_a x_{n-a}\big)^\top
\right\|_2
\le
\sum_{a=0}^{M}
\left(
\sum_{t=1}^{n}|e_r^\top\mathbf A_a x_{t-a}|^2
\right)^{1/2}.
\]
Since
\[
|e_r^\top \mathbf A_a x_{t-a}|
\le
\|\mathbf A_a\|_{\mathrm{op}}\|x_{t-a}\|_2,
\]
we obtain
\[
\|\mathbf G_{r,M}x\|_2
\le
\sum_{a=0}^{M}\|\mathbf A_a\|_{\mathrm{op}}
\left(
\sum_{t=1}^{n}\|x_{t-a}\|_2^2
\right)^{1/2}
\le
\left(\sum_{a=0}^{M}\|\mathbf A_a\|_{\mathrm{op}}\right)\|x\|_2
\le K_A \|x\|_2.
\]
Therefore,
\[
\|\mathbf G_{r,M}\|_{\mathrm{op}}\le K_A,
\qquad
\|\mathbf G_{s,M}\|_{\mathrm{op}}\le K_A.
\]
It follows that
\[
\|\mathbf L_{r,M}\|_{\mathrm{op}}
\le
\|\mathbf F_n\|_{\mathrm{op}}\|\mathbf G_{r,M}\|_{\mathrm{op}}
\le K_d K_A,
\]
and similarly
\[
\|\mathbf L_{s,M}\|_{\mathrm{op}}\le K_d K_A.
\]

Using the previous bounds,
\[
\|\mathbf Q_M\|_{\mathrm{op}}
\le
\|\mathbf L_{r,M}\|_{\mathrm{op}}
\|\mathbf H_n\|_{\mathrm{op}}
\|\mathbf L_{s,M}\|_{\mathrm{op}}
\le
C\ell_n,
\]
and
\[
\|\mathbf Q_M\|_F
\le
\|\mathbf L_{r,M}\|_{\mathrm{op}}
\|\mathbf H_n\|_F
\|\mathbf L_{s,M}\|_{\mathrm{op}}
\le
C\sqrt{n\ell_n}.
\]
All constants here depend only on $K_A,K_d,K_0$.

By (U1), the entries of $\boldsymbol\xi_M$ are independent, centered, and uniformly sub-Gaussian.
Therefore the Hanson--Wright inequality gives
\[
\Prob\!\left(
\left|
\boldsymbol\xi_M^\top \mathbf Q_M \boldsymbol\xi_M
-
\E(\boldsymbol\xi_M^\top \mathbf Q_M \boldsymbol\xi_M)
\right|
>x
\right)
\le
2\exp\!\left[
-c \min\!\left\{
\frac{x^2}{\|\mathbf Q_M\|_F^2},
\frac{x}{\|\mathbf Q_M\|_{\mathrm{op}}}
\right\}
\right].
\]
Using the norm bounds above and setting $x=nt$ yields
\[
\Prob\!\left(
\left|
{\hat{\mathbf{V}}^{or(M)}}[r,s]-\E{\hat{\mathbf{V}}^{or(M)}}[r,s]
\right|
>t
\right)
\le
2\exp\!\left[
-c_1 n \min\!\left\{
\frac{t^2}{\ell_n},
\frac{t}{\ell_n}
\right\}
\right],
\]
where $c_1>0$ depends only on $K_\varepsilon,K_A,K_d,K_0$, and crucially is
independent of $M,n,r,s$.

For each fixed $n$, since
\[
\sum_{a=0}^\infty \|\mathbf A_a\|_{\mathrm{op}}<\infty,
\]
we have
\[
\max_{1\le r\le p_n}\|Z_{t,r}-Z_{t,r}^{(M)}\|_2
\le
C \sum_{a>M}\|\mathbf A_a\|_{\mathrm{op}}
\to 0,
\qquad M\to\infty,
\]
uniformly in the finitely many $t$ appearing in the estimator.
By the definition of $\mathbf U_t$ and $\sum_j|d_{n,j}|\le K_d$,
\[
\max_{1\le r\le p_n}\|U_{t,r}-U_{t,r}^{(M)}\|_2\to0.
\]
Since $\Vor[r,s]$ is a finite weighted average of products
$U_{t,r}U_{u,s}$, it follows that
\[
{\hat{\mathbf{V}}^{or(M)}}[r,s]\to \Vor[r,s]
\quad\text{in }L^1,
\]
and hence
\[
\E{\hat{\mathbf{V}}^{or(M)}}[r,s]\to \E\Vor[r,s].
\]

Now fix $\tau\in(0,t)$. Then
\begin{align*}
\Prob\!\left(\big|\Vor[r,s]-\E\Vor[r,s]\big|>t\right)
\le
&\Prob\!\left(\big|{\hat{\mathbf{V}}^{or(M)}}[r,s]-\E{\hat{\mathbf{V}}^{or(M)}}[r,s]\big|>t-\tau\right)
\\
&+
\Prob\!\left(
\big|\Vor[r,s]-{\hat{\mathbf{V}}^{or(M)}}[r,s]\big|
+
\big|\E\Vor[r,s]-\E{\hat{\mathbf{V}}^{or(M)}}[r,s]\big|
>\tau
\right).
\end{align*}
Letting $M\to\infty$, the second probability tends to zero, while the first is bounded by the
truncated-process inequality. Hence
\[
\Prob\!\left(\big|\Vor[r,s]-\E\Vor[r,s]\big|>t\right)
\le
2\exp\!\left[
-c_1 n \min\!\left\{
\frac{(t-\tau)^2}{\ell_n},
\frac{t-\tau}{\ell_n}
\right\}
\right].
\]
Finally let $\tau\downarrow0$ to obtain
\[
\Prob\!\left(\big|\Vor[r,s]-\E\Vor[r,s]\big|>t\right)
\le
2\exp\!\left[
-c_1 n \min\!\left\{
\frac{t^2}{\ell_n},
\frac{t}{\ell_n}
\right\}
\right].
\]
This proves the proposition.
\end{proof}

\begin{proposition}
\label{prop:oracle_bias_unified}
Under Assumption~\ref{ass:suff_unified}, there exists a constant $C_B>0$,
independent of $n$ and $(r,s)$, such that
\[
\|\E(\Vor)-\Vtrue\|_{\max}
\le
C_B\left\{\ell_n^{-q_0}+\frac{(m_n+1)h_n}{n}\right\}.
\]
Consequently, Assumption~\ref{ass:oracle_bias} holds.
\end{proposition}

\begin{proof}
Let $\Gamma_k=\E(\mathbf Z_0\mathbf Z_k^\top)$ and
\[
\Gamma_k^U:=\E(\mathbf U_0\mathbf U_k^\top).
\]
By Lemma~\ref{lem:basic_unified},
\[
\|\Gamma_k\|_{\max}\le C_\Gamma (1+|k|)^{-(q_0+2+\eta)},
\qquad k\in\mathbb Z,
\]
and
\[
\sum_{k\in\mathbb Z}\|\Gamma_k\|_{\max}<\infty,
\qquad
\sum_{k\in\mathbb Z}(1+|k|)^{q_0}\|\Gamma_k\|_{\max}<\infty.
\tag{A.1}
\label{eq:Gamma_summable_unified}
\]
By stationarity of $\{\mathbf U_t\}$,
\[
\E\Goorhatk
=
\frac1n
\sum_{t=m_n h_n+|k|+1}^{n}
\E(\mathbf U_t\mathbf U_{t-|k|}^\top)
=
\left(1-\frac{m_n h_n+|k|}{n}\right)\Gamma_{|k|}^U.
\]
Therefore,
\[
\E\Vor
=
\sum_{|k|<\ell_n}
K\!\left(\frac{k}{\ell_n}\right)
\left(1-\frac{m_n h_n+|k|}{n}\right)
\Gamma_{|k|}^U.
\]
Recall
\[
\Vtrue=\sum_{k\in\mathbb Z}\Gamma_k.
\]
Add and subtract appropriate terms to write
\[
\E\Vor-\Vtrue
=
B_{\mathrm{ker}}+B_{\mathrm{tail}}+B_{\mathrm{bdry}}+B_{\mathrm{cross}},
\]
where
\begin{align*}
B_{\mathrm{ker}}
&:=
\sum_{|k|<\ell_n}
\left\{
K\!\left(\frac{k}{\ell_n}\right)-1
\right\}\Gamma_k,\\
B_{\mathrm{tail}}
&:=
-\sum_{|k|\ge\ell_n}\Gamma_k,\\
B_{\mathrm{bdry}}
&:=
-\sum_{|k|<\ell_n}
K\!\left(\frac{k}{\ell_n}\right)
\frac{m_n h_n+|k|}{n}\Gamma_k,\\
B_{\mathrm{cross}}
&:=
\sum_{|k|<\ell_n}
K\!\left(\frac{k}{\ell_n}\right)
\left(1-\frac{m_n h_n+|k|}{n}\right)
(\Gamma_k^U-\Gamma_k).
\end{align*}
Hence
\[
\|\E\Vor-\Vtrue\|_{\max}
\le
\|B_{\mathrm{ker}}\|_{\max}
+\|B_{\mathrm{tail}}\|_{\max}
+\|B_{\mathrm{bdry}}\|_{\max}
+\|B_{\mathrm{cross}}\|_{\max}.
\]
By Assumption~\ref{ass:kernel_diff},
\[
|1-K(x)|\le C_K |x|^{q_0},
\qquad |x|\le1.
\]
Therefore,
\[
\|B_{\mathrm{ker}}\|_{\max}
\le
\sum_{|k|<\ell_n}
\left|1-K\!\left(\frac{k}{\ell_n}\right)\right|
\|\Gamma_k\|_{\max}
\le
C_K \ell_n^{-q_0}
\sum_{|k|<\ell_n}|k|^{q_0}\|\Gamma_k\|_{\max}.
\]
Using \eqref{eq:Gamma_summable_unified},
\[
\|B_{\mathrm{ker}}\|_{\max}\lesssim \ell_n^{-q_0}.
\]
Similarly,
\[
\|B_{\mathrm{tail}}\|_{\max}
\le
\sum_{|k|\ge\ell_n}\|\Gamma_k\|_{\max}
\le
\ell_n^{-q_0}
\sum_{|k|\ge\ell_n}|k|^{q_0}\|\Gamma_k\|_{\max}
\lesssim \ell_n^{-q_0}.
\]
Using $\sup_{|x|\le1}|K(x)|\le K_0$,
\[
\|B_{\mathrm{bdry}}\|_{\max}
\le
\frac{K_0}{n}
\sum_{|k|<\ell_n}(m_n h_n+|k|)\|\Gamma_k\|_{\max}.
\]
Hence
\[
\|B_{\mathrm{bdry}}\|_{\max}
\le
\frac{K_0(m_n h_n+\ell_n)}{n}
\sum_{k\in\mathbb Z}\|\Gamma_k\|_{\max}
\lesssim
\frac{m_n h_n+\ell_n}{n}.
\]
Since $\ell_n\le h_n$ by (U4),
\[
\|B_{\mathrm{bdry}}\|_{\max}
\lesssim
\frac{(m_n+1)h_n}{n}.
\]
By definition,
\[
\mathbf U_t=\sum_{j=0}^{m_n}d_{n,j}\mathbf Z_{t-jh_n},
\]
so
\[
\Gamma_k^U
=
\sum_{j=0}^{m_n}\sum_{j'=0}^{m_n}
d_{n,j}d_{n,j'}\,\Gamma_{k+(j-j')h_n}.
\]
Using the normalization in \eqref{eq:diff_normalized}, the diagonal part reproduces $\Gamma_k$,
thus
\[
\Gamma_k^U-\Gamma_k
=
\sum_{\substack{0\le j,j'\le m_n\\ j\neq j'}}
d_{n,j}d_{n,j'}\,\Gamma_{k+(j-j')h_n}.
\]
Therefore,
\[
\|\Gamma_k^U-\Gamma_k\|_{\max}
\le
\sum_{j\neq j'} |d_{n,j}d_{n,j'}|\,
\|\Gamma_{k+(j-j')h_n}\|_{\max}.
\]
Since $\sum_j|d_{n,j}|\le K_d$,
\[
\|\Gamma_k^U-\Gamma_k\|_{\max}
\le
K_d^2
\max_{1\le |a|\le m_n}\|\Gamma_{k+a h_n}\|_{\max}.
\]
Now fix $|k|<\ell_n$ and $1\le |a|\le m_n$. Then
\[
|k+a h_n|\ge h_n-\ell_n.
\]
Hence, by Lemma~\ref{lem:basic_unified},
\[
\|\Gamma_{k+a h_n}\|_{\max}
\le
C_\Gamma (1+h_n-\ell_n)^{-(q_0+2+\eta)}
\lesssim
h_n^{-(q_0+2+\eta)},
\]
where we used $\ell_n=o(h_n)$.
Thus uniformly in $|k|<\ell_n$,
\[
\|\Gamma_k^U-\Gamma_k\|_{\max}
\lesssim
h_n^{-(q_0+2+\eta)}.
\]
Therefore,
\[
\|B_{\mathrm{cross}}\|_{\max}
\le
\sum_{|k|<\ell_n}
\left|K\!\left(\frac{k}{\ell_n}\right)\right|
\|\Gamma_k^U-\Gamma_k\|_{\max}
\lesssim
\ell_n h_n^{-(q_0+2+\eta)}.
\]
By (U4),
\[
\ell_n h_n^{-(q_0+1+\eta)}=o(\ell_n^{-q_0}),
\]
hence a fortiori
\[
\ell_n h_n^{-(q_0+2+\eta)}=o(\ell_n^{-q_0}).
\]
Therefore,
\[
\|B_{\mathrm{cross}}\|_{\max}=o(\ell_n^{-q_0}).
\]
Collecting the estimates from Steps 3--6 gives
\[
\|\E(\Vor)-\Vtrue\|_{\max}
\le
C_1\ell_n^{-q_0}
+
C_2\frac{(m_n+1)h_n}{n}
+
o(\ell_n^{-q_0}).
\]
For all sufficiently large $n$, the little-$o$ term can be absorbed into the first term, so
\[
\|\E(\Vor)-\Vtrue\|_{\max}
\le
C_B\left\{
\ell_n^{-q_0}
+
\frac{(m_n+1)h_n}{n}
\right\}
\]
for some constant $C_B>0$ independent of $n$ and $(r,s)$.
This completes the proof.
\end{proof}

\begin{proposition}
\label{prop:Rmu_unified_direct}
Under Assumption~\ref{ass:suff_unified}, Assumption~\ref{ass:cross_bernstein} holds. More precisely, there exists a constant \(c_3>0\),
depending only on \(K_\varepsilon,K_A,K_d,K_0\), such that, conditional on
\(\{\bm\mu_t\}_{t=1}^n\), for all sufficiently large \(n\) and all \(t>0\),
\begin{equation}
\Prob\!\left(
\maxnorm{\mathbf R_{\mu,n}} > t
\ \middle|\ \{\bm\mu_t\}_{t=1}^n
\right)
\le
4p_n^2 \exp\!\left(
-c_3 n \frac{t^2}{\ell_n^2 \bar M_{2,n}}
\right).
\label{eq:Rmu_unified_direct}
\end{equation}
Consequently,
\begin{equation}
\Prob\!\left(
\maxnorm{\mathbf R_{\mu,n}} > t
\ \middle|\ \{\bm\mu_t\}_{t=1}^n
\right)
\le
4p_n^2 \exp\!\left[
-c_3 n \min\!\left\{
\frac{t^2}{L_n^2 \bar M_{2,n}},
\frac{t}{L_n \bar M_{\infty,n}}
\right\}
\right].
\label{eq:Rmu_unified_direct_weakened}
\end{equation}
In particular, this verifies Assumption~\ref{ass:cross_bernstein}.
\end{proposition}

\begin{proof}
Fix any \(1\le r,s\le p_n\). As in the proof of Proposition~\ref{prop:bernstein_unified},
for \(M\ge0\), define the truncated linear process
\[
\mathbf Z_t^{(M)}:=\sum_{a=0}^{M}\mathbf A_a\boldsymbol\varepsilon_{t-a},
\qquad
\mathbf U_t^{(M)}:=\sum_{j=0}^{m_n} d_{n,j}\mathbf Z_{t-jh_n}^{(M)}.
\]
Let \(\mathbf R_{\mu,n}^{(M)}\) be the analogue of \(\mathbf R_{\mu,n}\) obtained by replacing
\(\mathbf U_t\) with \(\mathbf U_t^{(M)}\) in \eqref{eq:Rmu}. Then, by \eqref{eq:Rmu},
\begin{align}
\mathbf R_{\mu,n}^{(M)}[r,s]
&=
\sum_{|k|<\ell_n} K\!\left(\frac{k}{\ell_n}\right)
\frac{1}{n}
\sum_{t=m_n h_n+|k|+1}^{n}
\Bigl\{
M_t^{(r)} U_{t-|k|}^{(M,s)}
+
U_t^{(M,r)} M_{t-|k|}^{(s)}
\Bigr\}
\notag\\
&=: I_{r,s}^{(M)}+J_{r,s}^{(M)}.
\label{eq:RmuM_cross_decomp}
\end{align}

Introduce the deterministic vectors
\[
\mathbf m_r
:=
(0,\ldots,0,M_{m_n h_n+1}^{(r)},\ldots,M_n^{(r)})^\top\in\mathbb R^n,
\]
and similarly
\[
\mathbf m_s
:=
(0,\ldots,0,M_{m_n h_n+1}^{(s)},\ldots,M_n^{(s)})^\top\in\mathbb R^n.
\]
Then, using the same boundary convention as in the proof of Proposition~\ref{prop:bernstein_unified},
\begin{equation*}
I_{r,s}^{(M)}=\frac1n\,\mathbf m_r^\top \mathbf H_n \mathbf u_s^{(M)},
\qquad
J_{r,s}^{(M)}=\frac1n\,(\mathbf u_r^{(M)})^\top \mathbf H_n \mathbf m_s,
\end{equation*}
where \(\mathbf H_n\) is the kernel matrix introduced in that proof, and
\[
\mathbf u_q^{(M)}=(U_{1,q}^{(M)},\ldots,U_{n,q}^{(M)})^\top,\qquad q=r,s.
\]

Next, let
\[
\boldsymbol\xi_M
:=
(\boldsymbol\varepsilon_{1-M}^\top,\ldots,\boldsymbol\varepsilon_n^\top)^\top
\in\mathbb R^{(n+M)p_n},
\]
and recall from the proof of Proposition~\ref{prop:bernstein_unified} that there exist matrices
\(\mathbf L_{r,M}\) and \(\mathbf L_{s,M}\) such that
\[
\mathbf u_r^{(M)}=\mathbf L_{r,M}\boldsymbol\xi_M,
\qquad
\mathbf u_s^{(M)}=\mathbf L_{s,M}\boldsymbol\xi_M,
\]
with
\begin{equation}
\|\mathbf L_{r,M}\|_{\op}\le K_dK_A,
\qquad
\|\mathbf L_{s,M}\|_{\op}\le K_dK_A,
\qquad
\|\mathbf H_n\|_{\op}\le C\ell_n.
\label{eq:L_H_bounds_Rmu}
\end{equation}
Therefore,
\begin{equation*}
I_{r,s}^{(M)}
=
\frac1n\,\mathbf g_{r,s,M}^\top \boldsymbol\xi_M,
\qquad
\mathbf g_{r,s,M}:=\mathbf L_{s,M}^\top \mathbf H_n^\top \mathbf m_r,
\end{equation*}
and
\begin{equation*}
J_{r,s}^{(M)}
=
\frac1n\,\widetilde{\mathbf g}_{r,s,M}^\top \boldsymbol\xi_M,
\qquad
\widetilde{\mathbf g}_{r,s,M}:=\mathbf L_{r,M}^\top \mathbf H_n \mathbf m_s.
\end{equation*}

We now bound the Euclidean norms of these coefficient vectors. By the definition of
\(\bar M_{2,n}\) in \eqref{eq:Mbar_def},
\begin{equation*}
\|\mathbf m_r\|_2^2
\le
\sum_{t=m_n h_n+1}^{n}\norm{\mathbf M_t}_{\max}^2
=
n\bar M_{2,n},
\qquad
\|\mathbf m_s\|_2^2
\le
n\bar M_{2,n}.
\end{equation*}
Hence, by \eqref{eq:L_H_bounds_Rmu},
\begin{align}
\|\mathbf g_{r,s,M}\|_2
&\le
\|\mathbf L_{s,M}\|_{\op}\,\|\mathbf H_n\|_{\op}\,\|\mathbf m_r\|_2
\le
C\,\ell_n \sqrt{n\bar M_{2,n}},
\label{eq:g_bound}
\\
\|\widetilde{\mathbf g}_{r,s,M}\|_2
&\le
\|\mathbf L_{r,M}\|_{\op}\,\|\mathbf H_n\|_{\op}\,\|\mathbf m_s\|_2
\le
C\,\ell_n \sqrt{n\bar M_{2,n}}.
\label{eq:gtilde_bound}
\end{align}

Since the entries of \(\boldsymbol\xi_M\) are independent, centered, and uniformly sub-Gaussian by Assumption~\ref{ass:suff_unified}(U1), the standard concentration inequality for linear forms in
independent sub-Gaussian random variables yields
\[
\Prob\!\left(
\bigl|\mathbf v^\top \boldsymbol\xi_M\bigr|>x
\right)
\le
2\exp\!\left(
-c\frac{x^2}{\|\mathbf v\|_2^2}
\right)
\]
for every deterministic vector \(\mathbf v\), where \(c>0\) depends only on \(K_\varepsilon\).
Applying this with \(\mathbf v=\mathbf g_{r,s,M}\) and \(x=nt/2\), and using \eqref{eq:g_bound}, we obtain
\begin{equation}
\Prob\!\left(
|I_{r,s}^{(M)}|>\frac t2
\ \middle|\ \{\bm\mu_t\}_{t=1}^n
\right)
\le
2\exp\!\left(
-c\,n\frac{t^2}{\ell_n^2\bar M_{2,n}}
\right).
\label{eq:I_tail_direct}
\end{equation}
Similarly, by \eqref{eq:gtilde_bound},
\begin{equation}
\Prob\!\left(
|J_{r,s}^{(M)}|>\frac t2
\ \middle|\ \{\bm\mu_t\}_{t=1}^n
\right)
\le
2\exp\!\left(
-c\,n\frac{t^2}{\ell_n^2\bar M_{2,n}}
\right).
\label{eq:J_tail_direct}
\end{equation}

Combining \eqref{eq:RmuM_cross_decomp}, \eqref{eq:I_tail_direct}, and \eqref{eq:J_tail_direct},
we obtain
\begin{equation}
\Prob\!\left(
|\mathbf R_{\mu,n}^{(M)}[r,s]|>t
\ \middle|\ \{\bm\mu_t\}_{t=1}^n
\right)
\le
4\exp\!\left(
-c\,n\frac{t^2}{\ell_n^2\bar M_{2,n}}
\right).
\label{eq:RmuM_entry_tail}
\end{equation}

It remains to let \(M\to\infty\). For each fixed \(n\), since
\(\sum_{a=0}^\infty \|\mathbf A_a\|_{\op}<\infty\),
\[
\max_{1\le q\le p_n}\|Z_{t,q}-Z_{t,q}^{(M)}\|_2
\le
C\sum_{a>M}\|\mathbf A_a\|_{\op}\to0,
\qquad M\to\infty.
\]
Using \(\sum_{j=0}^{m_n}|d_{n,j}|\le K_d\), it follows that
\[
\max_{1\le q\le p_n}\|U_{t,q}-U_{t,q}^{(M)}\|_2\to0
\qquad M\to\infty.
\]
Since \(\mathbf R_{\mu,n}[r,s]\) is a finite weighted sum of terms of the form
\(M_t^{(r)}U_{t-|k|}^{(s)}\) and \(U_t^{(r)}M_{t-|k|}^{(s)}\), with
\(\{\bm\mu_t\}_{t=1}^n\) treated as deterministic conditioning information, we have
\[
\mathbf R_{\mu,n}^{(M)}[r,s]\to \mathbf R_{\mu,n}[r,s]
\qquad\text{in }L^1
\]
for each fixed \(r,s\). Letting \(M\to\infty\) in \eqref{eq:RmuM_entry_tail} gives
\begin{equation*}
\Prob\!\left(
|\mathbf R_{\mu,n}[r,s]|>t
\ \middle|\ \{\bm\mu_t\}_{t=1}^n
\right)
\le
4\exp\!\left(
-c\,n\frac{t^2}{\ell_n^2\bar M_{2,n}}
\right).
\end{equation*}

Finally, applying the union bound over all \(1\le r,s\le p_n\), we obtain
\begin{align*}
\Prob\!\left(
\maxnorm{\mathbf R_{\mu,n}}>t
\ \middle|\ \{\bm\mu_t\}_{t=1}^n
\right)
&\le
\sum_{r=1}^{p_n}\sum_{s=1}^{p_n}
\Prob\!\left(
|\mathbf R_{\mu,n}[r,s]|>t
\ \middle|\ \{\bm\mu_t\}_{t=1}^n
\right)
\notag\\
&\le
4p_n^2\exp\!\left(
-c_3 n\frac{t^2}{\ell_n^2\bar M_{2,n}}
\right),
\end{align*}
which proves \eqref{eq:Rmu_unified_direct}.

Since \(\ell_n\le L_n\), we have
\[
\frac{t^2}{L_n^2\bar M_{2,n}}
\le
\frac{t^2}{\ell_n^2\bar M_{2,n}},
\]
and therefore
\[
\min\!\left\{
\frac{t^2}{L_n^2\bar M_{2,n}},
\frac{t}{L_n\bar M_{\infty,n}}
\right\}
\le
\frac{t^2}{\ell_n^2\bar M_{2,n}}.
\]
Hence \eqref{eq:Rmu_unified_direct_weakened} follows immediately from
\eqref{eq:Rmu_unified_direct}, and therefore Assumption~\ref{ass:cross_bernstein} holds. This completes the proof.
\end{proof}

\bibliographystyle{apa}
\bibliography{ref}

\end{document}